\documentclass{llncs}

\usepackage{a4wide}

\usepackage{xcolor,pgf,tikz,pgflibraryarrows,pgffor,pgflibrarysnakes}
\usepackage{bbold}
\usepackage{lscape}
\usepackage{enumerate}

\usepackage{amsmath}
\usepackage{amssymb}
\usepackage{stmaryrd}

\usepackage{comment}
\usepackage{times}
\usepackage{url}

\usepackage{graphicx}
\usepackage{wrapfig}
\usepackage[dvips]{epsfig}
\usepackage{xspace}
\usepackage{color}
\usepackage[latin1]{inputenc}

\newcommand{\outcome}{{\sf{outcome}}}
\newcommand{\Outcome}{{\sf{Outcome}}}
\newcommand{\Strat}{\Pi}

\newcommand{\MP}{\textnormal{\textsf{MP}}}

\def\sg{\mathrel[\joinrel\mathrel[}
\def\sd{\mathrel]\joinrel\mathrel]}
\newcommand{\sem}[1]{\sg \mathrel{#1} \sd}

\newcommand{\size}[1]{|#1|}

\newcommand{\rationals}{{\mathbb{Q}}}

\newcommand{\integers}{{\mathbb{Z}}}
\newcommand{\nat}{{\mathbb{N}}}

\newcommand{\Val}{{\sf Val}}
\newcommand{\Plays}{{\sf Plays}}
\newcommand{\Pref}{{\sf Pref}}
\newcommand{\ValOpt}{\nu}

\newcommand{\LTL}{\textnormal{\textsf{LTL}}\xspace}
\newcommand{\LTLMP}{\textnormal{\textsf{LTL$_\textsf{MP}$}}\xspace}
\newcommand{\LTLE}{\textnormal{\textsf{LTL$_\textsf{E}$}}\xspace}
\newcommand{\Lit}{{\sf Lit}}

\newcommand{\nextLTL}{{\sf{X}}}
\newcommand{\until}{{\sf{U}}}
\newcommand{\always}{\Box}
\newcommand{\eventually}{\Diamond}

\newcommand{\true}{\ensuremath{\textsf{true}}\xspace}
\newcommand{\false}{\ensuremath{\textsf{false}}\xspace}

\newcommand{\uca}{\text{ucb}}
\newcommand{\nba}{\text{nb}}
\newcommand{\parity}{\text{dp}}

\newcommand{\EL}{{\sf EL}}
\newcommand{\parityO}{{\sf Parity}}
\newcommand{\energyO}{{\sf PosEnergy}}
\newcommand{\safetyO}{{\sf Safety}}
\newcommand{\MPO}{{\sf MeanPayoff}}

\newcommand{\DP}{{\sf DP}\xspace}
\newcommand{\NBW}{{\sf NB}\xspace}

\newcommand{\UCW}{{\sf UCB}\xspace}

\newcommand{\UCWK}{{\sf U\textit{K}CB}\xspace}

\newcommand{\final}{\alpha}

\newcommand{\infG}{{\sf{Inf}}}
\newcommand{\infA}{{\sf{Inf}}}
\newcommand{\runs}{{\sf{Runs}}}
\newcommand{\visit}{{\sf{Visit}}}

\newcommand{\Win}{{\sf{Win}}}
\newcommand{\energy}{{\sf{e}}}
\newcommand{\ConstantC}{C}

\begin{document}

\title{Synthesis from LTL Specifications with \\Mean-Payoff Objectives}

\author{Aaron Bohy$^1$ \and V\'eronique Bruy\`ere$^1$ \and Emmanuel Filiot$^2$ \and Jean-Fran\c{c}ois Raskin$^{3}$\thanks{Author supported by ERC Starting Grant (279499: inVEST).}}	 
\institute{$^1$Universit\'e de Mons$\quad$ $^2$Universit\'e Paris-Est Cr\'eteil$\quad$ $^3$Universit\'e
  Libre de Bruxelles}

  \maketitle	

\begin{abstract}
    The classical LTL synthesis problem is purely qualitative: the given LTL specification
    is realized or not by a reactive system. LTL is not expressive enough to formalize
    the correctness of reactive systems with respect to some quantitative aspects.
    This paper extends the \emph{qualitative} LTL synthesis setting to a \emph{quantitative} setting. The alphabet
    of actions is extended with a weight function ranging over the rational numbers. The value of an infinite word is the mean-payoff
    of the weights of its letters. The synthesis problem then amounts to automatically construct (if possible) a reactive system  whose executions 
all satisfy a given LTL formula and have mean-payoff values greater than or equal to some
given threshold. The latter problem is called \LTLMP synthesis and the \LTLMP realizability problem asks to check whether such a system exists. 
    We first show that \LTLMP realizability is not more difficult than LTL realizability: it is 
2ExpTime-Complete. This is done by reduction to two-player mean-payoff parity games. 
While infinite memory strategies are required to realize \LTLMP specifications
in general, we show that $\epsilon$-optimality can be obtained with finite memory strategies, for any $\epsilon > 0$.
To obtain an efficient algorithm in practice, we define a Safraless procedure to decide whether
there exists a finite-memory strategy that realizes a given specification for some
given threshold. This procedure is based on a reduction to two-player energy safety games which are in turn reduced to safety games. 
Finally, we show that those safety games
can be solved efficiently by exploiting the structure of their state spaces and by using antichains as a symbolic data-structure. 
All our results extend to multi-dimensional weights.
We have implemented an antichain-based procedure and we report on some promising experimental results.
\end{abstract}

\section{Introduction}

Formal specifications of reactive systems are usually expressed using formalisms like the linear temporal logic (LTL), the branching time temporal logic (CTL), or automata formalisms like B\"uchi automata. Those formalisms allow the specifier to express {\em Boolean properties} (called {\em qualitative properties} in the sequel) in the sense that a reactive system either conforms to them, or violates them. Additionally to those qualitative formalisms, there is a clear need for another family of formalisms that are able to express {\em quantitative properties} of reactive systems. Abstractly, a quantitative property can be seen as a function that maps an execution of a reactive system to a numerical value. For example, in a client-server application, this numerical value could be the mean number of steps that separate the time at which a request has been emitted by a client and the time at which this request has been granted by the server along an execution. Quantitative properties are concerned with a large variety of aspects like quality of service, bandwidth, energy consumption,... But quantities are also useful to compare the merits of alternative solutions, e.g. we may prefer a solution in which the quality of service is high and the energy consumption is low. Currently, there is a large effort of the research community with the objective to lift the theory of formal verification and synthesis from the {\em qualitative world} to the richer {\em quantitative world}~\cite{DBLP:conf/models/Henzinger12} (see related works for more details). In this paper, we consider mean-payoff and energy objectives. The alphabet of actions is extended with a weight function ranging over the rational numbers. A mean-payoff objective is a set of infinite words such that the mean value of the weights of their letters is greater than or equal to a given rational threshold \cite{DBLP:journals/tcs/ZwickP96}, while an energy objective
is parameterized by a non-negative initial energy level $c_0$ and contains all the words whose finite prefixes have a sum of weights greater than or equal to $-c_0$ \cite{DBLP:conf/formats/BouyerFLMS08}.

In this paper, we participate to this research effort by providing theoretical complexity results, practical algorithmic solutions, and a tool for the automatic synthesis of reactive systems from {\em quantitative specifications} expressed in the linear time temporal logic LTL extended with (multi-dimensional) mean-payoff and (multi-dimensional) energy objectives. To illustrate our contributions, let us consider the following specification of a controller that should grant exclusive access to a resource to two clients.
\begin{example} \label{ex:LTL} 
A client requests access to the resource by setting to true its request signal ($r_1$ for client~$1$ and $r_2$ for client~$2$), and the server grants those requests by setting to true the respective grant signal $g_1$ or $g_2$. We want to synthetize a server that eventually grants any client request, and that only grants one request at a time. This can be formalized in \LTL where the signals in $I = \{r_1, r_2\}$ are controlled by the environment (the two clients), and the signals in $O = \{g_1,w_1, g_2,w_2 \}$ are controlled by the server:
\begin{equation*} 
	\begin{split}
		\phi_1 &= \square(r_1 \rightarrow \nextLTL(w_1 \until g_1))\\
		\phi_2 &= \square(r_2 \rightarrow \nextLTL(w_2 \until g_2))\\
		\phi_3 &= \square(\lnot g_1 \vee \lnot g_2)\\
		\phi &= \phi_1 \wedge \phi_2 \wedge \phi_3
	\end{split} 
\end{equation*}
Intuitively, $\phi_1$ (resp. $\phi_2$) specifies that any request of client~$1$ (resp. client~$2$) must be eventually granted, and in-between
 the waiting signal $w_1$ (resp. $w_2$) must be high. Formula $\phi_3$ stands for mutual exclusion.

The formula $\phi$ is realizable. One possible strategy for the server is to alternatively assert $w_2,g_1$ and $w_1,g_2$, i.e. alternatively grant client~$1$ and client~$2$. While this strategy is formally correct, as it realizes the formula $\phi$ against all possible behaviors of the clients, it may not be the one that we expect. Indeed, we may prefer a solution that does not make unsollicited grants for example. Or, we may prefer a solution that gives, in case of request by both clients, some priority to client~$2$'s request. In the later case, one elegant solution would be to associate a cost equal to $2$ when $w_2$ is true and a cost equal to $1$ when $w_1$ is true. This clearly will favor solutions that give priority to requests from client~$2$ over requests from client~$1$. We will develop several other examples in the paper and describe the solutions that we obtain automatically with our algorithms.   
\end{example}

\paragraph{{\bf Contributions}}
We now detail our contributions and give some hints about the proofs. In Section~\ref{sec:problem}, we define the realizability problems for \LTLMP (\LTL extended with mean-payoff objectives) and \LTLE (\LTL extended with energy objectives), and give some examples. In Section \ref{sec:complexity}, we show that, as for the \LTL realizability problem, both the \LTLMP and \LTLE realizability problems are 2ExpTime-Complete.
As the proof of those three results follow a similar
structure, let us briefly recall how the 2ExpTime upper bound of the classical
\LTL realizability problem is established in~\cite{PnuRos89b}. The formula is first turned into
an equivalent nondeterministic B\"uchi automaton, which is then transformed into a deterministic parity automaton using Safra's construction. 
The latter automaton can be seen as a two-player
parity game in which Player~$1$ wins if and only if the formula is realizable. For the \LTLMP realizability problem,
our construction follows the same structure, except that we go to a two-player parity game with an
additional mean-payoff objective, and for the \LTLE realizability problem, we need to consider a parity game with an additional energy objective. 
By a careful analysis of the complexity of all the steps
involved in those two constructions, we build, on the basis of results in~\cite{CD10} and~\cite{CHJ05}, solutions that provide
the announced 2ExpTime upper bound. 


It is known that winning mean-payoff parity games may require infinite memory strategies, but there exist $\epsilon$-optimal finite-memory strategies~\cite{CHJ05}. In contrast, for energy parity games, it is known that finite-memory optimal strategies exist~\cite{CD10}. In Section~\ref{sec:complexity}, we show that those results transfer to \LTLMP (resp. \LTLE) realizability problems thanks to the reduction of these problems to mean-payoff (resp. energy) parity games. 
Furthermore, we show that under finite-memory strategies, \LTLMP realizability is 
in fact equivalent to \LTLE realizability: a specification is \MP-realizable under
finite-memory strategies if and only if it is ${\sf E}$-realizable, by simply shifting the weights of the signals by the threshold value. 
Because finite-memory strategies are more interesting in practice, we thus concentrate on the \LTLE realizability problem in the rest of the paper. 

Even if recent progresses have been made \cite{TsaiFVT10}, Safra's construction is intricate and notoriously difficult to implement efficiently~\cite{ATW06}.
We develop in Section~\ref{sec:algorithm}, following~\cite{DBLP:conf/focs/KupfermanV05}, a Safraless procedure for the \LTLE realizability problem, that is based on a reduction to a safety game, with the nice property to transform a quantitative objective into a simple qualitative objective.
The main building blocks of this procedure are as follows. 
\emph{(1)} Instead of transforming an \LTL formula
into a deterministic parity automaton, 
we prefer to use a universal co-B\"uchi automaton as proposed in~\cite{DBLP:conf/focs/KupfermanV05}. To deal with the energy objectives, we thus transform the formula into a universal co-B\"uchi energy automaton for some initial credit $c_0$, which requires that all runs on an input word $w$ visit finitely many accepting states and the energy level of $w$ is always positive starting from the initial energy level $c_0$. 
\emph{(2)}  By strenghtening the co-B\"uchi condition into a $K$-co-B\"uchi condition (as done in~\cite{ScheweF07a,FMSD11}), where
at most $K$ accepting states can be visited by each run, we then go to an
energy safety game. We show that for sufficiently large value $K$ and initial credit $c_0$, this reduction is complete.%
\emph{(3)} Any energy safety game is equivalent to a safety game, as shown in~\cite{DBLP:journals/fmsd/BrimCDGR11}.

Finally, in Section \ref{sec:Implementation}, we discuss some implementation issues. 
The proposed Safraless construction has two main advantages. Firstly, the search for winning
strategies for \LTLE realizability can be incremental on $K$ and $c_0$ 
(avoiding in practice to consider the large theoretical bounds $\mathbb{K}$ and $\mathbb{C}$ that ensure completeness). 
Secondly, the state space of the safety game can be partially ordered and solved by a backward fixpoint algorithm. Since the latter manipulates sets of states closed for this order, it can be made efficient and symbolic by working only on the antichain of their maximal elements. As described in Section~\ref{sec:multi-dim}, our results can be extended to the multi-dimensional case, i.e. tuples of weights.
All the algorithms have been implemented in our tool {\sf Acacia+}~\cite{DBLP:conf/cav/BohyBFJR12}, and promising experimental results are reported in Section \ref{sec:exp}. 

\paragraph{{\bf Related works}}
The \LTL synthesis problem has been first solved in~\cite{PnuRos89b}, Safraless approaches have been proposed in~\cite{JobstmannB06,DBLP:conf/focs/KupfermanV05,ScheweF07a,FMSD11}, and implemented in prototypes of tools~\cite{JobstmannB06,Cav09,Ehlers-CAV-2010,DBLP:conf/cav/BohyBFJR12}. All those works only treat plain qualitative \LTL, and not the quantitative extensions considered in this article.

Mean-payoff games~\cite{DBLP:journals/tcs/ZwickP96} and energy games~\cite{DBLP:conf/formats/BouyerFLMS08,DBLP:journals/fmsd/BrimCDGR11}, extensions with parity conditions~\cite{CHJ05,CD10,BouyerMOU11}, or multi-dimensions~\cite{DBLP:conf/fsttcs/ChatterjeeDHR10,CRR12} have recently received a large attention from the research community. The use of such game formalisms has been advocated in~\cite{DBLP:conf/cav/BloemCHJ09} for specifying quantitative properties of reactive systems. Several among the motivations developed in~\cite{DBLP:conf/cav/BloemCHJ09} are similar that our motivations for considering quantitative extensions of \LTL.
All these related works make the assumption that the game graph is given explicitly (and not implicitly using an \LTL formula), as in our case. 

In~\cite{DBLP:conf/lics/BokerCHK11}, Boker et al. introduce extensions of linear and branching time temporal logics with operators to express constraints on values accumulated along the paths of a weighted Kripke structure. One of their extensions is similar to \LTLMP. However the authors of~\cite{DBLP:conf/lics/BokerCHK11} only study the complexity of model-checking problems whereas we consider realizability and synthesis problems.
\section{Problem statement}\label{sec:problem}

\subsection{Preliminaries}

\paragraph{Linear temporal logic --} 
The formulas of linear temporal logic (\LTL) are defined over
a finite set $P$ of atomic propositions. The syntax is given by the
grammar:
$$
\phi\ ::=\ p\ |\ \phi\vee\phi\ |\ \neg\phi\ |\ \nextLTL \phi\ |\ \phi\until\phi \qquad
p\in P
$$

\noindent The notations \true, \false, $\phi_1\wedge \phi_2$,  $\eventually
\phi$ and $\always \phi$ are defined as usual. \LTL formulas $\phi$ are interpreted on infinite words
$u = \sigma_0\sigma_1\sigma_2\dots \in (2^P)^\omega$ via a satisfaction
relation $u \models \phi$ inductively defined as: 

\begin{tabular}{ll}
$u \models p$ & if $p\in\sigma_0$ \\
$u \models \phi_1\vee\phi_2$ & if $u\models \phi_1$ or $u\models \phi_2$ \\
$u \models \neg \phi$ & if $u\not\models\phi$ \\
$u \models \nextLTL \phi$ & if $\sigma_1\sigma_2\dots\models \phi$ \\
$u \models \phi_1\until\phi_2$ & if  $\exists n\geq 0, \sigma_n\sigma_{n+1}\dots\models \phi_2$ and  $\forall i, 0\leq i< n,
\sigma_i\sigma_{i+1}\dots\models \phi_1$.
\end{tabular}

\noindent
Given an \LTL formula $\phi$, we denote by $\sem{\phi}$ the set of words $u$ such that $u \models \phi$.

\paragraph{\LTL Realizability and synthesis --}
The realizability problem for \LTL is
best seen as a game between two players. Let  $\phi$ be an \LTL formula over the set $P = I \uplus O$ partitioned into $I$ the set of {\em input  signals} controlled by Player $I$ (the environment), and $O$ the set of {\em output signals}
 controlled by Player $O$ (the controller). With this partition of $P$, we associate
the three following alphabets: $\Sigma_P=2^P$, $\Sigma_O=2^O$,
and $\Sigma_I=2^I$. 

The realizability game is played in turns. Player~$O$ starts by giving
$o_0 {\in} \Sigma_O$, Player~$I$ responds by giving
$i_0 {\in} \Sigma_I$, then Player~$O$ gives $o_1 {\in} \Sigma_O$ and
Player~$I$ responds by $i_1 {\in} \Sigma_I$, and so on. This game lasts forever
and the outcome of the game is the infinite word $(o_0 \cup
i_0)(o_1\cup i_1)(o_2\cup i_2)\dots \in \Sigma_P^{\omega}$.

The players play according to \emph{strategies}. A strategy for Player~$O$ is
a  mapping $\lambda_O:(\Sigma_O\Sigma_I)^*\rightarrow
\Sigma_O$, while a strategy for Player~$I$ is a mapping
$\lambda_I : (\Sigma_O\Sigma_I)^*\Sigma_O\rightarrow \Sigma_I$. The
outcome of the strategies $\lambda_O$ and $\lambda_I$ is the word
$\Outcome(\lambda_O,\lambda_I) = (o_0\cup i_0)(o_1\cup i_1)\dots$
such that $o_0 =\lambda_O(\epsilon)$,
$i_0 = \lambda_I(o_0)$ and for all $k \geq 1$, $o_k = \lambda_O(o_0i_0\dots
o_{k-1}i_{k-1})$ and $i_k = \lambda_I(o_0i_0\dots o_{k-1}i_{k-1}o_k)$.
We denote by $\Outcome(\lambda_O)$ the set of all outcomes 
$\Outcome(\lambda_O,\lambda_I)$ with $\lambda_I$ any strategy of Player~$I$.
We let $\Strat_O$ (resp. $\Strat_I$) be the set of strategies for Player~$O$ (resp. Player~$I$).

Given an \LTL formula $\phi$ (the specification), the {\em \LTL realizability problem} is to
decide whether there exists a strategy $\lambda_O$ of Player~$O$ such that 
$\Outcome(\lambda_O,\lambda_I)\models \phi$
against all strategies $\lambda_I$ of Player~$I$. If such a \emph{winning} strategy exists, we say that
the specification $\phi$ is \textit{realizable}. The \textit{\LTL synthesis problem} asks
to produce a strategy $\lambda_O$ that realizes $\phi$, when it is realizable.

\paragraph{Moore machines --}
It is known that the \LTL realizability problem is 2ExpTime-Complete and that finite-memory strategies suffice in case of realizability \cite{PnuRos89b}. A strategy $\lambda_O$ of Player~$O$ is \emph{finite-memory} if there exists a right-congruence $\sim$ on $(\Sigma_O\Sigma_I)^*$ of finite index such that $\lambda_O(u) = \lambda_O(u')$ for all $u \sim u' $. It is equivalent to say that it can be described by a \emph{Moore machine} ${\cal M} = (M,m_0,\alpha_{U},\alpha_{N})$ defined as follows. The non-empty set $M$ is the finite memory\footnote{The memory $M$ is the set of equivalence classes for $\sim$.} of ${\cal M}$ and $m_0$ is its initial memory state. The memory update function $\alpha_U : M \times \Sigma_I \rightarrow M$ modifies the current memory state at each $i \in \Sigma_I$ emitted by Player~$I$, and the next-move function $\alpha_N : M \rightarrow \Sigma_O$ indicates which $o \in \Sigma_O$ is proposed by Player~$O$ given the current memory state. The function $\alpha_U$ is naturally extended to words $u \in \Sigma_I^*$. The \emph{language} of ${\cal M}$, denoted by $L({\cal M})$, is the set of words $u = (o_0\cup i_0)(o_1\cup i_1)\dots \in \Sigma_P^\omega$ such that $o_0 = \alpha_N(m_0)$ and for all $k\geq 1$, $o_k = \alpha_N(\alpha_U(m_0,i_0\dots i_{k-1}))$.  The size $\size{\cal M}$ of a Moore machine is defined as the size $|M|$ of its memory. Therefore, with these notations, an \LTL formula is realizable iff there exists a Moore machine such that $L({\cal M}) \subseteq \ \sem{\phi}$.

\begin{theorem}[\cite{PnuRos89b}]
 \label{thm:LTLRealizability}
The \LTL realizability problem is 2ExpTime-Complete and any realizable \LTL formula is realizable by a finite-memory strategy.
\end{theorem}

\subsection{Synthesis with mean-payoff objectives}

\paragraph{\LTLMP realizability and synthesis --} 
Consider a finite set $P$ partitioned as $I \uplus O$. Let $\Lit(P)$ be the set $\{p \mid p \in P\} \cup \{\neg p \mid p \in P\}$ of literals over $P$, and let $w : \Lit(P) \rightarrow \integers$ be a \emph{weight function} where positive numbers represent rewards\footnote{\label{fnrational}We use weights at several places of this paper. In some statements and proofs, we take the freedom to use \emph{rational} weights as it is equivalent up to rescaling. However we always assume that weights are integers encoded in binary for complexity results.}. This function is extended to $\Sigma_I$ (resp. $\Sigma_O$) as follows:  $w(i) = \Sigma_{p \in i} w(p) +  \Sigma_{p \in I \setminus \{i\}} w(\neg p)$ for $i \in \Sigma_I$ (resp. $w(o) = \Sigma_{p \in o} w(p) +  \Sigma_{p \in O \setminus \{o\}} w(\neg p)$ for $o \in \Sigma_O$). In this way, it can also be extended to $\Sigma_P$ as $w(o \cup i) = w(o) + w(i)$ for all $o \in \Sigma_O$ and $i \in \Sigma_I$.\footnote{\label{fndecomp}The decomposition of $w(o \cup i)$ as the sum $w(o) + w(i)$ emphasizes the partition of $P$ as $I \uplus O$ and will be useful in some proofs.} In the sequel, we denote by $\langle P, w \rangle$ the pair given by the finite set $P$ and the weight function $w$ over $\Lit(P)$; we also use the weighted alphabet $\langle \Sigma_P, w \rangle$.

Consider an \LTL formula $\phi$ over $\langle P, w \rangle$ and an outcome $u = (o_0\cup i_0)(o_1\cup i_1) \dots \in \Sigma_P^\omega$ produced by Players~$I$ and $O$. We associate a \emph{value} $\Val(u)$ with $u$ that captures the two objectives of Player~$O$ of both satisfying $\phi$ and achieving a mean-payoff objective. For each $n \geq 0$, let $u(n)$ be the prefix of $u$ of length $n$. We define the \emph{energy level} of $u(n)$ as $\EL(u(n)) = \sum_{k=0}^{n-1} w(o_k) + w(i_k)$. We then assign to $u$ a \emph{mean-payoff value} equal to $\MP(u)=\liminf_{n \rightarrow \infty} \frac{1}{n} \EL(u(n))$. Finally we define the value of $u$ as:

\begin{center}
$\Val(u) =$
$\begin{cases} \MP(u) \text{ if } u \models \phi\\  - \infty \text{ otherwise.} \end{cases}$
\end{center}

Given an \LTL formula $\phi$ over $\langle P, w \rangle$ and a threshold $\nu \in \rationals$, the \emph{\LTLMP realizability problem (resp. \LTLMP realizability problem under finite memory)} asks to decide whether there exists a strategy (resp. finite-memory strategy) $\lambda_O$ of Player~$O$ such that $\Val(\Outcome(\lambda_O,\lambda_I)) \geq \nu$ against all strategies $\lambda_I$ of Player~$I$, 
in which we say that $\phi$ is \emph{\MP-realizable (resp. \MP-realizable under finite memory)} . The \emph{\LTLMP synthesis problem} is to produce such a winning strategy $\lambda_O$ for Player~$O$. Therefore the aim is to achieve two objectives: \emph{(i)} realizing $\phi$, \emph{(ii)}  having a long-run average reward greater than the given threshold.

\paragraph{Optimality --}
Given $\phi$ an \LTL formula over $\langle P, w \rangle$, the \emph{optimal value} (for Player~$O$)  is defined as 
$$\ValOpt_{\phi} = \sup_{\lambda_O \in \Strat_O} \inf_{\lambda_I \in \Strat_I} \Val(\Outcome(\lambda_O,\lambda_I)).$$
For a real-valued $\epsilon \geq 0$, a strategy $\lambda_O$ of Player~$O$ is \emph{$\epsilon$-optimal} if $\Val(\Outcome(\lambda_O,\lambda_I)) \geq \ValOpt_{\phi} - \epsilon$ against all strategies $\lambda_I$ of Player~$I$. It is \emph{optimal} if it is $\epsilon$-optimal with $\epsilon = 0$. Notice that $\ValOpt_{\phi}$  is equal to $-\infty$ if Player~$O$ cannot realize $\phi$. 

\begin{example} \label{ex:LTLMP}
Let us come back to Example \ref{ex:LTL} of a client-server system with two clients sharing a resource. The specification have been formalized by an \LTL formula $\phi$ over the alphabet $P = I \uplus O$, with $I = \{r_1, r_2\}$, $O = \{g_1, w_1, g_2, w_2\}$.
Suppose that we want to add the following 
constraints: client~$2$'s requests take the priority over client 1's requests, but
client~$1$'s should still be eventually granted. Moreover, we would like to keep
minimal the delay between requests and grants. This latter requirement has more the flavor of an optimality criterion and is best modeled using a weight function and a mean-payoff objective. To this end, we impose penalties to the waiting signals $w_1,w_2$ controlled by 
Player $O$, with a larger penalty to $w_2$ than to $w_1$. 
We thus use the following weight function $w : \Lit(P) \rightarrow \integers$:
\begin{center}
$w(l) =$
$\begin{cases} -1 \text{ if } l = w_1\\ -2 \text{ if } l = w_2\\ 0 \text{ otherwise.} \end{cases}$
\end{center}

One optimal strategy for the server is to behave as follows: it almost always grants the resource to client~2 immediately after $r_2$ is set to true by client~2, and with a decreasing frequency grants request $r_1$ emitted by client 1. Such a server ensures a mean-payoff value equal to $-1$ against the most demanding behavior of the clients (where they are constantly requesting the shared resource). Such a strategy requires the server to use an infinite memory as it has to grant client 1 with an infinitely decreasing frequency. Note that a server that would grant client 1 in such a way without the presence of requests by client 1 would still be optimal.

It is easy to see that no finite memory server can be optimal. Indeed, if we allow the server to count only up to a fixed positive integer $k$ $\in$ $\mathbb{N}$ then the best that this server can do is as follows: grant immediatly any request by client 2 if the last ungranted request of client 1 has been emitted less than $k$ steps in the past, otherwise grant the request of client 1. The mean-payoff value of this solution, in the worst-case (when the two clients always emit their respective request) is equal to $-(1 + \frac{1}{k})$. So, even if finite memory cannot be optimal in this example, it is the case that given any $\epsilon > 0$, we can devise a finite-memory strategy that is $\epsilon$-optimal. 
\end{example}

\paragraph{\LTLE realizability and synthesis --}

For the proofs of this paper, we also need to consider realizability and synthesis with \emph{energy objectives} (instead of mean-payoff objectives). With the same notations as before, the \emph{\LTLE realizability problem} is to decide whether $\phi$ is \emph{{\sf E}-realizable}, that is, whether there exists a  strategy $\lambda_O$ of Player~$O$ and an integer $c_0 \in \nat$ such that for all strategies $\lambda_I$ of Player~$I$, \emph{(i)} $u = \Outcome(\lambda_O,\lambda_I)\models \phi$, \emph{(ii)} $\forall n \geq 0, c_0 +  \EL(u(n)) \geq 0$. 
Instead of requiring that $\MP(u) \geq \nu$ for some given theshold $\nu$ as a second objective, we thus ask if there exists an \emph{initial credit} $c_0$ such that the energy level of each prefix $u(n)$ remains positive. When $\phi$ is {\sf E}-realizable, the \emph{\LTLE synthesis problem} is to produce such a winning strategy $\lambda_O$ for Player~$O$. Finally, we define the \emph{minimum initial credit} as the least value of initial credit for which $\phi$ is {\sf E}-realizable. A strategy $\lambda_O$ is \emph{optimal} if it is winning for the minimum initial credit.

\section{Computational complexity of the \LTLMP realizability problem}
\label{sec:complexity}

In this section, we solve the \LTLMP realizability problem, and we establish its complexity. Our solution relies on a reduction to a mean-payoff parity game. The same result also holds for the \LTLE realizability problem.

\begin{theorem} \label{thm:complexity}
The \LTLMP realizability problem is 2ExpTime-Complete.
\end{theorem}

Before proving this result, we recall useful notions on parity automata and on game graphs.

\subsection{Deterministic parity automata}
A \emph{deterministic parity automaton} over a finite alphabet $\Sigma$ is a tuple
${\cal A} = (\Sigma, Q, q_0,  \delta,p)$ where $Q$ is a finite set of states with $q_0$  the initial state, $\delta : Q \times \Sigma \rightarrow Q$ is a transition function that assigns a unique state\footnote{In this definition, a deterministic parity automaton is also complete} to each given state and symbol, and $p : Q \rightarrow \nat$ is a priority function that assigns a priority to each state. 

For infinite words $u = \sigma_0\sigma_1 \dots \in \Sigma^{\omega}$, there exists a unique run $\rho(u) = \rho_0\rho_1\dots\in Q^\omega$ such that
$\rho_0 = q_0$ and $\forall k \geq 0, \rho_{k+1} = \delta(\rho_k,\sigma_k)$. We denote by $\infA(\rho(u))$ the set of states that appear infinitely often in $\rho(u)$. The \emph{language} $L({\cal A})$ is the set of words $u \in \Sigma^{\omega}$ such that $\min \{p(q) \mid q \in \infA(\rho(u)) \}$ is even. We have the next theorem (see for instance \cite{DBLP:conf/focs/KupfermanV05}).

\begin{theorem} \label{thm:LTLParity}
Let $\phi$ be an \LTL formula over $P$.  One can construct a deterministic parity automaton ${\cal A}_{\phi}$ such that $L({\cal A}_{\phi})$ $=$ $\sem{\phi}$. If $\phi$ has size $n$, then ${\cal A}_\phi$ has $2^{2^{O(n \log n)}}$ states and $2^{O(n)}$ priorities. 
\end{theorem}

\subsection{Game graphs} \label{subsec:gamegraphs}

A {\em game graph}  $G = (S,s_0,E)$ consists of a finite set $S$ of states partitioned into $S_1$ the states of Player 1, and $S_2$ the states of Player 2 (that is $S = S_1 \uplus S_2$), an initial state $s_0$, and a set $E \subseteq S \times S$ of edges such that for all $s \in S$, there exists a state $s' \in S$ such that $(s,s') \in E$. A game on $G$ starts from the initial state $s_0$ and is played in rounds as follows. If the game is in a state belonging to $S_1$, then Player~1 chooses the successor state among the set of outgoing edges; otherwise Player~2 chooses the successor state. Such a game results in a {\em play} that is an infinite path $\rho=s_0s_1\dots s_n \dots$, whose prefix $s_0s_1\dots s_n$ of length\footnote{The length is counted as the number of edges.} $n$ of is denoted by $\rho(n)$. 
We denote by $\Plays(G)$ the set of all plays in $G$ and by $\Pref(G)$ the set of all prefixes of plays in $G$.
A \emph{turn-based} game is a game graph $G$ such that $E \subseteq (S_1 \times S_2) \cup (S_2 \times S_1)$, meaning that each game is played in rounds alternatively by Player~1 and Player~2.

\paragraph{Objectives --}

An {\em objective} for $G$ is a set $\Omega \subseteq S^{\omega}$. Let $p : S \rightarrow \nat$ be a \emph{priority function} and $w : E \rightarrow \integers$ be a \emph{weight function} where positive weights represent rewards.
%
%
The {\em energy level} of a prefix $\gamma=s_0s_1\dots s_n$ of a play is $\EL_G(\gamma)=\sum_{i=0}^{n-1} w(s_i, s_{i+1})$, and the {\em mean-payoff value} of a play $\rho=s_0s_1\dots s_n \dots$ is $\MP_G(\rho)= \liminf_{n \rightarrow \infty} \frac{1}{n} \cdot \EL_G(\rho(n))$.\footnote{Notation $\EL$, $\MP$ and $\Outcome$ is here used with the index $G$ to avoid any confusion with the same notation introduced in the previous section.} Given a play $\rho$, we denote $\infG(\rho)$ the set of states $s \in S$ that appear infinitely often in $\rho$. The following objectives $\Omega$ are considered in the sequel:
 
 \begin{itemize}
  	\item {\em Safety objective}. Given a set $\alpha \subseteq S$, the safety objective is defined as $\safetyO_G(\alpha) = \Plays(G) \cap \alpha^{\omega}$.
  	\item {\em Parity objective}. The parity objective is defined as $\parityO_G(p) = \{ \rho \in \Plays(G) \mid \min \{ p(s) \mid s \in \infG(\rho) \} \mbox{~is~even}\}$.
	\item {\em Energy objective}. Given an initial credit $c_0 \in \nat$, the energy objective is defined as $\energyO_G(c_0) = \{ \rho \in \Plays(G) \mid \forall n \geq 0: c_0 + \EL_G(\rho(n)) \geq 0 \}$.
	\item {\em Mean-payoff objective}. Given a threshold $\nu \in \rationals$, the mean-payoff objective is defined as $\MPO_{G}(\nu)$ $=$ $\{ \rho \in \Plays(G) \mid \MP_G(\rho) \geq \nu \}$.
	\item {\em Combined objective}. The  {\em energy safety} objective $ \energyO_{G}(c_0) \cap \safetyO_G(\alpha)$ (resp.  {\em energy parity} objective $\energyO_{G}(c_0) \cap \parityO_G(p)$, {\em mean-payoff parity} objective $\MPO_{G}(\nu) \cap \parityO_G(p)$) combines the requirements of energy and safety (resp. energy and parity, energy and mean-payoff) objectives.
\end{itemize}

When an objective $\Omega$ is imposed on a game $G$, we say that $G$ is an \emph{$\Omega$ game}. For instance, if $\Omega$ is an energy safety objective, we say that $G$ is an \emph{energy safety game}, aso.

\paragraph{Strategies --}

Given a game graph $G$, a \emph{strategy} for Player~1 is a function $\lambda_1 : S^*S_1 \rightarrow S$ such that $(s,\lambda_1(\gamma \cdot s)) \in E$ for all $\gamma \in S^*$ and $s \in S_1$.
A play $\rho=s_0 s_1 \dots s_n \dots$ starting from the initial state $s_0$ is compatible with $\lambda_1$ if for all $k \geq 0$ such that $s_k \in S_1$ we have $s_{k+1}=\lambda_1(\rho(k))$. 
Strategies and play compatibility are defined symmetrically for Player~2. The set of strategies of Player~1 (resp. Player~2) is denoted by $\Strat_1$ (resp. $\Strat_2$). We denote by $\Outcome_G(\lambda_1,\lambda_2)$ the play from $q_0$, called outcome, that is compatible with $\lambda_1$ and $\lambda_2$. The set of all outcomes $\Outcome_G(\lambda_1,\lambda_2)$, with $\lambda_2$ any strategy of Player~2, is denoted by  $\Outcome_G(\lambda_1)$.
A strategy $\lambda_1$ for Player~1 is {\em winning} for an objective $\Omega$ if $\Outcome_G(\lambda_1) \subseteq \Omega$. We also say that $\lambda_1$ is winning in the $\Omega$ game $G$.  

A strategy $\lambda_1$ of Player~1 is \emph{finite-memory} if there exists a right-congruence $\sim$ on $\Pref(G)$ with finite index such that $\lambda_1(\gamma \cdot s_1) = \lambda_1(\gamma'  \cdot s_1)$ for all $\gamma \sim \gamma'$ and $s_1 \in S_1$. The \emph{size} of the memory is equal to the number of equivalence classes of $\sim$. We say that $\lambda_1$ is \emph{memoryless} if $\sim$ has only one equivalence class. In other words, $\lambda_1$ is a mapping $S_1 \rightarrow S$ that only depends on the current state.

\paragraph{Energy safety games --}

Let us consider a safety game $\langle G, \alpha \rangle$ with the safety objective $\safetyO_G(\alpha)$, or equivalently, with the objective to avoid $S \setminus \alpha$. The next classical fixpoint algorithm allows one to check whether Player~1 has a winning strategy (see \cite{GTW2002} for example). We define the fixpoint $\Win_1(\alpha)$ of the sequence $W_0 = \alpha$, $W_{k+1} = W_k \cap \{ \{s \in S_1 \mid \exists (s,s') \in E, s' \in W_k \} \cup \{s \in S_2 \mid \forall (s,s') \in E, s' \in W_k \}\}$ for all $k \geq 0$. It is well-known that Player~1 has a winning strategy $\lambda_1$ in the safety game $\langle G, \alpha \rangle$ iff $s_0 \in \Win_1(\alpha)$, and that the set $\Win_1(\alpha)$ can be computed in polynomial time. Moreover, the subgraph $G'$ of $G$ induced by $\Win_1(\alpha)$ is again a game graph (i.e. every state has an outgoing edge), and if $s_0 \in \Win_1(\alpha)$, then $\lambda_1$ can be chosen as a memoryless strategy $S_1 \rightarrow S$ such that $(s,\lambda_1(s))$ is an edge in $G'$ for all $s \in \Win_1(\alpha)$ (Player~1 forces to stay in $G'$). With this induced subgraph, we have the next reduction of energy safety games to energy games.

\begin{proposition} \label{prop:ESreductionE}
Let $\langle G, w, \alpha \rangle$ be an energy safety game. Let  $\langle G', w' \rangle$ be the energy game such that $G'$ is the subgraph of $G$ induced by $\Win_1(\alpha)$ and $w'$ is the restriction of $w$ to its edges. Then the winning strategies of Player~1 in $\langle G', w' \rangle$ are the winning strategies in $\langle G, w, \alpha \rangle$ that are restricted to the states of \ $\Win_1(\alpha)$. 
\qed\end{proposition}


For an energy game $\langle G, w \rangle$ (resp. energy safety game $\langle G, w, \alpha \rangle$), the \emph{initial credit problem} asks whether there exist an initial credit $c_0 \in \nat$ and a winning strategy for Player~1 for the objective $\energyO_{G}(c_0)$ (resp. $\energyO_{G}(c_0) \cap \safetyO(\alpha)$). It is known that for energy games, this problem can be solved in $\sf{NP}  \cap \sf{coNP}$, and memoryless strategies suffice to witness
the existence of winning strategies for Player $1$ \cite{DBLP:conf/formats/BouyerFLMS08,DBLP:conf/emsoft/ChakrabartiAHS03}. %
%
%
%
Moreover, if we store in the states of the game the current energy level up to some bound
$C\geq 0$, one gets a safety game (whose safe states are those states with a positive energy level). For a sufficiently large bound $C$, this safety game is equivalent to 
the initial energy game \cite{DBLP:journals/fmsd/BrimCDGR11}.  Intuitively, the states of this safety game are pairs $(s,c)$ with $s$ a state of $G$ and $c$ an energy level in ${\cal C} = \{\bot,0,1, \ldots \ConstantC\}$ (with $\bot < 0$). When adding a positive (resp. negative) weight to an energy level, we bound the sum to $\ConstantC$ (resp. $\bot = -1$). The safety objective is to avoid states of the form $(s,\bot)$. Formally, given $\langle G, w \rangle$ an energy game with $G = (S,s_0,E)$, we define the safety game $\langle G_{\ConstantC}, \alpha \rangle$ with a graph $G_{\ConstantC} = (S',s'_0,E')$ and a safety objective $\alpha$ as follows:

\begin{itemize}
\item $S' = \{ (s,c) \mid s \in S, c \in {\cal C} \}$
\item $s'_0 = (s_0,C)$
\item $((s,c),(s',c')) \in E'$ \quad if $e=(s,s') \in E$ and $c \oplus w(e) = c'$
\item $\alpha = \{ (s,c) \mid s \in S, c \neq \bot  \}$
\end{itemize}
In this definition, we use the operator $\oplus : {\cal C}\times \integers \rightarrow {\cal C}$ such that $c \oplus k = \min(\ConstantC, c+k)$ if $\{ c \neq \bot$ and $c+k \geq 0 \}$, and $\bot$ otherwise.

By Proposition~\ref{prop:ESreductionE}, it follows that energy safety games can be reduced to safety games.

\begin{theorem} \label{thm:ESreductionS}
Let $\langle G, w, \alpha \rangle$ be an energy safety game. Then one can construct a safety game $\langle G', \alpha' \rangle$ such that Player~1 has a winning strategy in $\langle G, w, \alpha \rangle$ iff he has a winning strategy in $\langle G', \alpha' \rangle$.
\qed\end{theorem}




\paragraph{Energy parity games and mean-payoff parity games --}

Given an energy parity game $\langle G, w, p \rangle$, we can also formulate the \emph{initial credit problem} as done previously for energy games.

\begin{theorem}[\cite{CD10}] \label{thm:EParity}
The initial credit problem for a given energy parity game $\langle G, w, p \rangle$ can be solved in time $O(|E| \cdot d \cdot |S|^{d+3} \cdot W)$ where $|E|$ is the number of edges of $G$, $d$ is the number of priorities used by $p$ and $W$ is the largest weight (in absolute value) used by $w$. Moreover if Player~1 wins, then he has a finite-memory winning strategy  with a memory size bounded by $4 \cdot |S| \cdot d \cdot W$. 
\end{theorem}

Let us turn to mean-payoff parity games $\langle G, w, p \rangle$. With each play $\rho \in \Plays(G)$, we associate a \emph{value} $\Val_G(\rho)$ defined as follows (as done in the context of \LTLMP realizability):
$$
\Val_G(\rho) = \left\{\begin{array}{llllll}
\MP_G(\rho) & \text{ if  } \rho \in \parityO_G(p)\\
-\infty & \text{ otherwise. } \\
\end{array}\right.
$$

The \emph{optimal value} for Player~1 is defined as 
$$\ValOpt_{G} = \sup_{\lambda_1 \in \Strat_1} \inf_{\lambda_2 \in \Strat_2} \Val_G(\Outcome_G(\lambda_1,\lambda_2)).$$

For a real-valued $\epsilon \geq 0$, a strategy $\lambda_1$ for Player~$1$ is \emph{$\epsilon$-optimal} if $\Val_G(\Outcome_G(\lambda_1,\lambda_2)) \geq \ValOpt_G - \epsilon$ against all strategies $\lambda_2$ of Player~$2$. It is \emph{optimal} if it is $\epsilon$-optimal with $\epsilon = 0$. 
If Player~1 cannot achieve the parity objective, then $\ValOpt_{G} = -\infty$, otherwise optimal strategies exist~\cite{CHJ05} and $\ValOpt_{G}$ is the largest threshold $\nu$ for which Player~1 can hope to achieve $\MPO_{G}(\nu)$.

\begin{theorem}[\cite{CHJ05,BouyerMOU11,CRR12}] \label{thm:MPparity}
The optimal value of a mean-payoff parity game $\langle G, w, p \rangle$ can be computed in time $O(  |E| \cdot |S|^{d+2} \cdot W)$. When $\ValOpt_{G} \neq -\infty$, optimal strategies for Player~1 may require infinite memory; however for all $\epsilon > 0$ Player~1 has a finite-memory $\epsilon$-optimal strategy.
\end{theorem}

\subsection{Reduction to a mean-payoff parity game}

\paragraph{Solution to the \LTLMP realizability problem --}
We can now proceed to the proof of Theorem~\ref{thm:complexity}. It is based on the following proposition.

\begin{proposition} \label{prop:reductionMPParity}
Let $\phi$ be an \LTL formula over $\langle P, w_P \rangle$. Then one can construct a mean-payoff parity game $\langle G_{\phi},w,p \rangle$ with $2^{2^{O(n \log n)}}$ states and $2^{O(n)}$ priorities such that the following are equivalent: for each threshold $\nu \in \rationals$
\begin{enumerate}
\item there exists a (finite-memory) strategy $\lambda_O$ of Player~$O$ such that $\Val(\Outcome(\lambda_O,\lambda_I)) \geq \nu$ against all strategies $\lambda_I$ of Player~$I$;
\item there exists a (finite-memory) strategy $\lambda_1$ of Player~$1$  such that $\Val_{G_\phi}(\Outcome_{G_\phi}(\lambda_1,\lambda_2)) \geq \nu$ against all strategies $\lambda_2$ of Player~$2$, in the game $\langle G_{\phi},w,p \rangle$. 
\end{enumerate}
Moreover, if $\lambda_1$ is a finite-memory strategy with size $m$, then $\lambda_O$ can be chosen as a finite-memory strategy with size $m \cdot |S_1|$ where $S_1$ is the set of states of Player~1 in $G_{\phi}$.\footnote{A converse of this corollary could also be stated but is of no utility in the next results.}

\end{proposition}

\begin{proof} 
Let $\phi$ be an \LTL formula over $P = I \uplus O$, and let $w_P : \Lit(P) \rightarrow \integers$ be a weight function. We first construct a deterministic parity automaton ${\cal A}_{\phi} = (\Sigma_P, Q, q_0, \delta,p)$ such that $L({\cal A}_{\phi}) = \sem{\phi}$ (see Theorem~\ref{thm:LTLParity}). This automaton has $2^{2^{O(n \log n)}}$ states and $2^{O(n)}$ priorities.

We then derive from ${\cal A}_{\phi}$ a turn-based mean-payoff parity game $\langle G_{\phi},w,p' \rangle$ with $G_{\phi} = (S,s_0,E)$ as follows. The initial state $s_0$ is equal to $(q_0,j_0)$ for some $j_0 \in I$.\footnote{Symbol $j_0$ can be chosen arbitrarily.} To get the turn-based aspect, the set $S$ is partionned as $S_1 \uplus S_2$ such that $S_1 = \{(q,i) \mid q \in Q, i \in \Sigma_I \}$ and $S_2 = \{(q,o) \mid q \in Q, o \in \Sigma_O \}$.\footnote{We will limit the set of states to the accessible ones.}  Let us describe the edges of $G_{\phi}$. For each $q \in Q$, $o \in \Sigma_O$ and $i \in \Sigma_I$, let $q' = \delta(q, o \cup i)$. Then $E$ contains the two edges $((q,j),(q,o))$ and $((q,o),(q',i))$ for all $j \in \Sigma_I$. We clearly have $E \subseteq (S_1 \times S_2) \cup (S_2 \times S_1)$. Moreover, since ${\cal A}_{\phi}$ is deterministic, we have the nice property that there exists a bijection $\Theta : \Sigma_P^* \rightarrow \Pref(G_{\phi}) \cap (S_1 S_2)^* S_1 $ defined as follows.  For each $u = (o_0 \cup i_0) (o_1 \cup i_1) \dots (o_n \cup i_n) \in \Sigma_P^*$, we consider in ${\cal A}_{\phi}$ the run $\rho(u) = q_0q_1\dots q_{n+1} \in Q^*$ such that $q_{k+1} = \delta(q_k,o_k \cup i_k)$ for all $k$. We then define $\Theta(u)$ as $(q_0,j_0)(q_0,o_0) (q_1,i_0) (q_1,o_1) (q_2,i_1) \cdots (q_{n},o_n)(q_{n+1},i_n)$. Clearly, $\Theta$ is a bijection and it can be extended to a bijection $\Theta : \Sigma_P^{\omega} \rightarrow \Plays(G_{\phi})$.

The priority function $p' : S \rightarrow \nat$ for $G_{\phi}$ is defined from the priority function $p$ of ${\cal A}_{\phi}$ by $p'(q,o) = p'(q,i) = p(q)$ for all $q \in Q, o \in O$ and $i \in I$. The weight function $w : E \rightarrow \integers$ for $G_{\phi}$ is defined from the weight function $w_P$ as follows. For all edges $e$ ending in a state $(q,o)$ with $o \in O$ (resp. $(q,i)$ with $i \in I$), we define $w(e) = w_P(o)$ (resp. $w(e) = w_P(i))$. Notice that $\Theta$ preserves the energy level, the meanpayoff value and the parity objective, since for each $u \in \Sigma_P^{\omega}$, we have \emph{(i)} $\EL(u(n)) = \EL_{G_\phi}(\Theta(u(n)))$ for all $n$\footnote{For this equality, it is useful to recall footnote~\ref{fndecomp}.}, \emph{(ii)} $\MP(u) = \MP_{G_{\phi}}(\Theta(u))$, and \emph{(iii)} $u \models \phi$ iff $\Theta(u)$ satisfies the objective $\parityO_{G_{\phi}}(p')$.

It is now easy to prove that the two statements of Proposition~\ref{prop:reductionMPParity} are equivalent. Suppose that \emph{1.} holds. Given the winning strategy $\lambda_O : (\Sigma_O\Sigma_I)^* \rightarrow \Sigma_O$, we are going to define a winning strategy $\lambda_1 : (S_1 S_2)^*S_1 \rightarrow S_2$ of Player~$1$ in $G_{\phi}$, with the idea that $\lambda_1$ mimics $\lambda_O$ thanks to the bijection $\Theta$. More precisely, 
for any prefix $\gamma \in (S_1 S_2)^*S_1$ compatible with $\lambda_1$, we let $\lambda_1(\gamma) = (q,o)$ such that  $(q,i_n)$ is the last state of $\gamma$ and $o = \lambda_O(o_0 i_0 o_1 i_1 \dots o_n i_n)$ with $(o_0 \cup i_0) (o_1 \cup i_1) \dots (o_n \cup i_n) = \Theta^{-1}(\gamma)$. In this way, $\lambda_1$ is winning. Indeed $\Val_{G_\phi}(\Outcome_{G_\phi}(\lambda_1,\lambda_2)) \geq \nu$ for any strategy $\lambda_2$ of Player~2 because $\lambda_O$ is winning and by \emph{(ii)} and \emph{(iii)} above. Moreover $\lambda_O$ is finite-memory iff $\lambda_1$ is finite-memory.

Suppose now that \emph{2.} holds. Given the (finite-memory) winning strategy $\lambda_1$, we define a (finite-memory) winning strategy $\lambda_O$ with the same kind of arguments as done above. 
We just give the definition of $\lambda_O$ from $\lambda_1$. We let $\lambda_O(o_0 i_0 o_1 i_1 \dots o_n i_n) = o$ such that $\lambda_1(\gamma) = (q,o)$ with $\gamma = \Theta((o_0 \cup i_0) (o_1 \cup i_1) \dots (o_n \cup i_n))$.

Suppose now that $\lambda_1$ is finite-memory with size $m$. Let $\sim_1$ be a right-congruence on $\Pref(G_{\phi})$ with index $m$ such that $\lambda_1(\gamma\cdot s) = \lambda_1(\gamma' \cdot s)$ for all $\gamma \sim_1 \gamma'$ and $s \in S_1$. To show that $\lambda_O$ is finite-memory, we have to define a right-congruence $\sim_O$ on $(\Sigma_O\Sigma_I)^*$ with finite index such that $\lambda_O(u) = \lambda_O(u')$ for all $u \sim_O u'$. Let $u = o_0 i_0 \dots o_n i_n, u' = o'_0 i'_0  \dots o'_l i'_l \in (\Sigma_O\Sigma_I)^*$, let $\Theta((o_0 \cup i_0) \dots (o_n \cup i_n)) = \gamma \cdot s, \Theta((o'_0 \cup i'_0) \dots (o'_l \cup i'_l)) = \gamma' \cdot s'$. Looking at the definition of $\lambda_O$, we see that $\sim_O$ has to be defined such that $u \sim_O u'$ if $\gamma \sim_1 \gamma'$ and $s = s'$. Moreover $\sim_O$ has index $m \cdot |S_1|$.

This completes the proof.\footnote{\label{fnsmaller}Notice that we could have defined a smaller game graph $G_{\phi}$ with $S_1 = Q$ 
and $w((q,o),q') = \max\{w_P(i) \mid \delta(q,o\cup i) = q'\}$ for all $(q,o) \in S_2, q' \in S_1$. The current definition simplifies the proof.} 
\qed\end{proof}

\begin{proof}[of Theorem~\ref{thm:complexity}]
By Proposition~\ref{prop:reductionMPParity}, solving the \LTLMP realizability problem is equivalent to checking whether Player~1 has a winning strategy in the mean-payoff parity game $\langle G_{\phi},w,p \rangle$ for the given threshold $\nu$. By Theorem~\ref{thm:MPparity}, this check can be done in time $O( |E| \cdot |S|^{d+2}  \cdot W)$. Since $G_{\phi}$ has $2^{2^{O(n \log n)}}$ states and $2^{O(n)}$ priorities (see Proposition~\ref{prop:reductionMPParity}), the \LTLMP realizability problem is in $O(2^{2^{O(n \log n)}} W)$.

This proves the 2Exptime-easyness of \LTLMP realizability problem. The 2Exptime-hardness of this problem is a consequence of the 2Exptime-hardness of \LTL realizability problem (see Theorem~\ref{thm:LTLRealizability}).
\qed\end{proof}

Proposition~\ref{prop:reductionMPParity} and its proof lead to the next two interesting corollaries. The first corollary is immediate. The second one asserts that the optimal value can be approached with finite memory strategies.

\begin{corollary} \label{cor:optimal}
Let $\phi$ be an \LTL formula and $\langle G_{\phi}, w, p \rangle$ be the associated mean-payoff parity game graph. Then $\ValOpt_{\phi} = \ValOpt_{G_{\phi}}$.  Moreover, when $\ValOpt_{\phi} \neq -\infty$, one can construct an optimal strategy $\lambda_1$ for Player~1 from an optimal strategy $\lambda_O$ for Player~$O$, and conversely. 
\qed\end{corollary}

\begin{corollary} \label{cor:LTLMPFiniteMem}
Let $\phi$ be an \LTL formula. If $\phi$ is $\MP$-realizable, then for all $\epsilon > 0$, Player~$O$ has an $\epsilon$-optimal winning strategy that is finite-memory, that is 
$$\ValOpt_{\phi} = \sup_{\lambda_O \in \Strat_O \atop \lambda_O \textnormal{ finite-memory }} \inf_{\lambda_I \in \Strat_I} \Val(\Outcome(\lambda_O,\lambda_I)).$$\end{corollary}

\begin{proof}
Suppose that $\phi$ is $\MP$-realizable. By Corollary~\ref{cor:optimal}, $\ValOpt_{\phi} = \ValOpt_{G_{\phi}} \neq - \infty$. Therefore, by Proposition~\ref{prop:reductionMPParity} and Theorem~\ref{thm:MPparity}, for each $\epsilon > 0$, Player~$1$ has a finite-memory winning strategy $\lambda_1$ in the mean-payoff parity $G_{\phi} = (S, s_0,E)$ for the threshold $\ValOpt_{\phi} - \epsilon$. By Proposition~\ref{prop:reductionMPParity}, from $\lambda_1$, we can derive a finite-memory winning strategy $\lambda_O$ for the \LTLMP realizability of $\phi$ for this threshold, which is thus the required finite-memory $\epsilon$-optimal winning strategy.
\qed\end{proof}

\paragraph{Solution to the \LTLE realizability problem --}
The same kind of arguments show that the \LTLE realizability problem is 2Exptime-complete. Indeed, in Proposition~\ref{prop:reductionMPParity}, it is enough to use a reduction with the same game $\langle G_{\phi}, w, p \rangle$, however with energy parity objectives instead of mean-payoff parity objectives. The proof is almost identical by taking the same initial credit $c_0$ for both the \LTLE realizability of $\phi$ and the energy objective in $G_{\phi}$, and by using Theorem~\ref{thm:EParity} instead of Theorem~\ref{thm:MPparity}.

\begin{theorem} \label{thm:complexityEnergy}
The \LTLE realizability problem is 2ExpTime-Complete.
\qed\end{theorem}

The next proposition states that when $\phi$ is ${\sf E}$-realizable, then Player~$O$ has a finite-memory strategy the size of which is related to $G_{\phi}$. This result is stronger than Corollary~\ref{cor:LTLMPFiniteMem} since it also holds for optimal strategies and it gives a bound on the memory size of the winning strategy.

\begin{proposition} \label{prop:LTLEFiniteMem}
Let $\phi$ be an \LTL formula over $\langle P, w_P \rangle$ and $\langle G_{\phi}, w, p \rangle$ be the associated energy parity game. Then $\phi$ is ${\sf E}$-realizable iff it is ${\sf E}$-realizable under finite memory. Moreover Player~$O$ has a finite-memory winning strategy with a memory size bounded by $4 \cdot |S|^2 \cdot d \cdot W$, where $|S|$ is the number of states of $G_{\phi}$, $d$ its number of priorities and $W$ its largest absolute weight.
\end{proposition}

\begin{proof}
By Proposition~\ref{prop:reductionMPParity} where $G_{\phi} = (S, s_0,E)$ is considered as an energy parity game, we know that Player~$1$ has a winning strategy $\lambda_1$ for the initial credit problem in $G_{\phi}$. Moreover, by Theorem~\ref{thm:EParity}, we can suppose that this strategy has finite-memory $M$ with $|M| \leq 4 \cdot |S| \cdot d \cdot W$. Finally, one can derive a finite-memory winning strategy $\lambda_O$ for the \LTLE realizabilty of $\phi$ with a memory size bounded by $|M| \cdot |S|$ by Proposition~\ref{prop:reductionMPParity}.
\qed\end{proof}

The constructions proposed in Theorems~\ref{thm:complexity} and~\ref{thm:complexityEnergy} can be easily extended to the more general case where the weights assigned to executions are given by a deterministic weighted automaton, as proposed in~\cite{ChatterjeeDH10}, instead of a weight function $w$ over $\Lit(P)$ as done here. Indeed, given an \LTL formula $\phi$ and a deterministic weighted automaton $\cal A$, we first construct from $\phi$ a deterministic parity automaton and then take the synchronized product with $\cal A$. Finally this product can be turned into a mean-payoff (resp. energy) parity game.

\section{Safraless algorithm} \label{sec:algorithm}

In the previous section, we have proposed an algorithm for solving the  \LTLMP realizability of a given \LTL formula $\phi$, which is based on a reduction to the mean-payoff parity game $G_{\phi}$. This algorithm has two main drawbacks. First, it requires the use of Safra's construction to get a deterministic parity automaton ${\cal A}_{\phi}$ such that $L({\cal A}_{\phi})$ $=$ $\sem{\phi}$, a construction which is resistant to efficient implementations~\cite{ATW06}. Second, strategies for the game $G_{\phi}$ may require infinite memory (for the threshold $\nu_{G_{\phi}}$, see Theorem~\ref{thm:MPparity}). This can also be the case for the $\LTLMP$ realizability problem, as illustrated by Example~\ref{ex:LTLMP}. In this section, we show how to circumvent these two drawbacks. 

\subsection{Finite-memory strategies}

The second drawback (infinite memory strategies) has been already partially solved by Corollary~\ref{cor:LTLMPFiniteMem}, when the threshold given for the \LTLMP-realizability is the optimal value $\nu_{\phi}$. Indeed it states the existence of finite-memory winning strategies for the thresholds $\nu_{\phi} - \epsilon$, for all $\epsilon > 0$. We here show that we can go further by translating the \LTLMP realizability problem under finite memory into an \LTLE realizability problem, and conversely. 

Recall (see Proposition~\ref{prop:reductionMPParity}) that testing whether an \LTL formula $\phi$ is \MP-realizable for a given threshold $\nu$ is equivalent to solve the mean-payoff parity game $G_{\phi}$ for the same threshold. Moreover, with the same game graph, testing whether $\phi$ is ${\sf E}$-realizable is equivalent to solve the energy parity game $G_{\phi}$. Let us study in more details the finite-memory winning strategies of $G_{\phi}$ seen either as a mean-payoff parity game, or as an energy parity game, through the next property proved in~\cite{CD10}.
%
%

\begin{proposition}[\cite{CD10}]\label{prop:Doyen}
Let $G = (S,s_0,E)$ be a game with a priority function $p$ and a weight function $w$. Let $\nu$ be a threshold and $w-\nu$ be the weight function such that $(w-\nu)(e) = w(e) -\nu$ for all edges $e$ of $G$. Let $\lambda_1$ be a finite-memory strategy for Player~1. Then $\lambda_1$ is winning in the mean-payoff parity game $\langle G, w, p \rangle$ with threshold $\nu$ iff $\lambda_1$ is winning in the energy parity game $\langle G, w-\nu, p \rangle$ for some initial credit~$c_0$.
\end{proposition}

This proposition leads to the next theorem.

\begin{theorem}\label{thm:MP-realIFFE-real}
Let $\phi$ be an \LTL formula $\phi$ over $\langle P, w_P \rangle$, and $\langle G_{\phi}, w, p \rangle$ be its associated mean-payoff parity game. Then
\begin{itemize}
\item the formula $\phi$ is \MP-realizable under finite memory for threshold $\nu$ iff $\phi$ over $\langle P, w_P - \nu \rangle$ is ${\sf E}$-realizable
\item if $\phi$ is \MP-realizable under finite memory, Player~$O$ has a winning strategy whose memory size is bounded by $4 \cdot |S|^2 \cdot d \cdot W$, where $|S|$ is the number of states of $G_{\phi}$, $d$ is the number of priorities of $p$ and $W$ is the largest absolute weight of the weight function $w - \nu$.
\end{itemize}
\end{theorem}

\begin{proof}
If $\phi$ is $\MP$-realizable under finite-memory for threshold $\nu$, then by Proposition~\ref{prop:reductionMPParity}, Player~1 has a finite-memory strategy $\lambda_1$ in the mean-payoff parity game $\langle G_{\phi}, w, p \rangle$. 
By Proposition~\ref{prop:Doyen}, $\lambda_1$ is a winning strategy in the energy parity game $\langle G_{\phi}, w-\nu, p \rangle$, that shows (by Proposition~\ref{prop:reductionMPParity}) that $\phi$ is \LTLE-realizable with weight function $w_P - \nu$ over $\Lit(P)$. The converse is proved similarly. 

Finally, given a winning strategy for the \LTLE-realizability of $\phi$ with weight function $w-\nu$ over $\Lit(P)$, we can suppose by Proposition~\ref{prop:LTLEFiniteMem} that it is finite-memory with a memory size bounded by $4 \cdot |S|^2 \cdot d \cdot W$, where $S$, $d$ and $W$ are the parameters of the energy parity game $\langle G_{\phi}, w-\nu, p \rangle$. This concludes the proof.
\qed\end{proof}

%

The next corollary follows from Theorem~\ref{thm:MP-realIFFE-real} and Corollary~\ref{cor:LTLMPFiniteMem}.

\begin{corollary} \label{cor:LTLMP - LTLE}
Let $\phi$ be an \LTL formula over $\langle P, w_P \rangle$. Then for all $\epsilon \in \rationals$, with $\epsilon > 0$, the following are equivalent:
\begin{enumerate}
\item $\lambda_O$ is a finite-memory $\epsilon$-optimal winning strategy for the \LTLMP-realizability of $\phi$ 
\item $\lambda_O$ is a winning strategy for the \LTLE-realizability of $\phi$ with weight function $ w_P - \ValOpt_{\phi} + \epsilon$ over $\Lit(P)$.
\end{enumerate}
\end{corollary}

%

It is important to notice that in this corollary, the memory size of the strategy (as described in Theorem~\ref{thm:MP-realIFFE-real}) increases as $\epsilon$ decreases. Indeed, it depends on the weight function $w - \ValOpt_{\phi} + \epsilon$ used by the energy parity game $G_{\phi}$. We recall that if $\epsilon = \frac{a}{b}$, then this weight function must be multiplied by $b$ in a way to have integer weights (see footnote~\ref{fnrational}). The largest absolute weight $W$ is thus also multiplied by $b$. 

In the sequel, to avoid strategies with infinite memory when solving the \LTLMP realizability problem, we will restrict to the \LTLMP realizability problem under finite memory. By Theorem~\ref{thm:MP-realIFFE-real}, it is enough to study winning strategies for the \LTLE realizability problem (having in mind that the weight function over $\Lit(P)$ has to be adapted). In the sequel, we only study this problem.

\subsection{Safraless construction}


To avoid the Safra's construction needed to obtain a deterministic parity automaton for the underlying \LTL formula, we adapt a Safraless construction proposed in~\cite{FMSD11,ScheweF07a} for the \LTL synthesis problem, in a way to deal with weights and efficiently solve the  \LTLE synthesis problem.  Instead of constructing a mean-payoff parity game from a deterministic parity automaton as done in Proposition~\ref{prop:reductionMPParity}, we will propose a reduction to a safety game. In this aim, we need to define the notion of energy automaton.

\paragraph{Energy automata --} 
Let  $\langle P,w\rangle$ with $P$ a finite set of signals and $w$ a weight function over \Lit$(P)$. We are going to recall several notions of automata on infinite words over $\Sigma_P$ and introduce the related notion of energy automata over the weighted alphabet $\langle \Sigma_P, w \rangle$. An \emph{automaton} ${\cal A}$ over the alphabet $\Sigma_P$ is a tuple $(\Sigma_P, Q, q_0, \final, \delta)$ such that $Q$ is a finite set of states, $q_0 \in Q$ is the initial state, $\final \subseteq Q$ is a set of final states and $\delta : Q \times \Sigma_P \rightarrow 2^Q$ is a transition function. We say that ${\cal A}$ is \textit{deterministic} if $\forall q\in Q, \forall \sigma\in\Sigma_P, |\delta(q,\sigma)|\leq 1$. It is \textit{complete}
if $\forall q\in Q, \forall \sigma\in\Sigma_P, \delta(q,\sigma)\neq\varnothing$.

A {\em run} of ${\cal A}$ on a word $u = \sigma_0 \sigma_1 \dots \in\Sigma_P^\omega$ is an infinite sequence of states $\rho = \rho_0\rho_1\dots\in Q^\omega$ such that $\rho_0 = q_0$ and $\forall k \geq 0, \rho_{k+1}\in \delta(\rho_k, \sigma_k)$. We denote by $\runs_{\cal A}(u)$ the set of runs of ${\cal A}$ on $u$, and by $\visit(\rho,q)$ the number of times the state $q$ occurs along the run $\rho$. We consider the following acceptance conditions:  

\begin{center}
\begin{tabular}{ll}
\emph{Non-deterministic B\"uchi}: & \quad  $\exists  \rho\in\runs_{\cal A}(u),  \exists q\in\final, \visit(\rho,q)=\infty$ \\
\emph{Universal co-B\"uchi}: &\quad  $\forall \rho\in\runs_{\cal A}(u), \forall q\in\final, \visit(\rho,q)<\infty$ \\
\emph{Universal $K$-co-B\"uchi}:&\quad  $\forall \rho\in\runs_{\cal A}(u), \sum_{q\in\final} \visit(\rho,q)\leq K$. 
\end{tabular}
\end{center}

A word $u \in \Sigma_P^{\omega}$ is \emph{accepted} by a \emph{non-deterministic B\"uchi automaton} (\NBW) ${\cal A}$ if $u$ satisfies the non-deterministic B\"uchi acceptance condition. We denote by $L_{\nba}({\cal A})$ the set of words accepted by ${\cal A}$. Similarly we have the notion of \emph{universal co-B\"uchi automaton} (\UCW) $\cal A$ (resp. \emph{universal $K$-co-B\"uchi automaton} (\UCWK) $\langle {\cal A}, K \rangle$) and the set  $L_{\uca}({\cal A})$ (resp. $L_{\uca,K}({\cal A})$) of accepted words. 

We also now introduce the energy automata. Let $\cal A$ be a \NBW over the alphabet $\Sigma_P$. The related \emph{energy non-deterministic B\"uchi automaton} (\energy\NBW) ${\cal A}^w$ is over the weighted alphabet $\langle \Sigma_P,w\rangle$ and has the same structure as $\cal A$. Given an initial credit $c_0 \in \nat$, a word $u$ is \emph{accepted} by ${\cal A}^w$ if  \emph{(i)} $u$ satisfies the non-deterministic B\"uchi acceptance condition and \emph{(ii)} $\forall n \geq 0, c_0 +  \EL(u(n)) \geq 0$. We denote by $L_{\nba}({\cal A}^w ,c_0)$ the set of words accepted by ${\cal A}^w$ with the given initial credit $c_0$. We also have the notions of \emph{energy universal co-B\"uchi automaton} (\energy\UCW) ${\cal A}^w$ and \emph{energy universal $K$-co-B\"uchi automaton} (\energy\UCWK) $\langle {\cal A}^w, K \rangle$, and the related sets  $L_{\uca}({\cal A}^w, c_0)$ and $L_{\uca,K}({\cal A}^w,c_0)$. Notice that if $K \leq K'$, then $L_{\uca,K}({\cal A}^w,c_0) \subseteq L_{\uca,K'}({\cal A}^w,c_0)$, and that if $c_0 \leq c'_0$, then $L_{\uca,K}({\cal A}^w,c_0) \subseteq L_{\uca,K}({\cal A}^w,c'_0)$.

The interest of \UCWK is that they can be determinized with the subset construction extended with counters \cite{FMSD11,ScheweF07a}. This construction also holds for \energy\UCWK. Intuitively, for all states $q$ of $\cal A$, we count (up to $\top = K+1$) the maximal number of accepting states which have been visited by runs ending in $q$. This counter is equal to $-1$ when no run ends in $q$. The final states are the subsets in which a state has its counter greater than $K$ (accepted runs will avoid them). Formally, let ${\cal A}^w$ be a \UCWK $(\langle \Sigma_P, w\rangle, Q, q_0, \final, \delta)$ with $K \in \nat$. With ${\cal K} = \{-1,0\dots,K,\top\}$ (with $\top > K$), we define $\det({\cal A}^w,K) = (\langle \Sigma_P, w\rangle, {\cal F},F_0,\beta,\Delta)$ where:
\begin{itemize}
\item $ {\cal F} = \{ F\ |\ F\text{ is a mapping from } Q  \text{ to } {\cal K} \}$
\item  $F_0 = q\in Q \mapsto \left\{\begin{array}{ll} -1 & \text{if }
  q\neq q_0 \\
  (q_0\in \final) & \text{otherwise} \end{array}\right.$
\item $\beta  = \{ F \in {\cal F} \mid \exists q, F(q) =  \top \}$
\item  $\Delta(F,\sigma)  = q\mapsto \text{max} \{F(p) \oplus (q\in\final)\ \mid q\in \delta(p,\sigma)\}$
\end{itemize}
In this definition, $(q\in\final)=1$ if $q$ is in $\final$, and $0$ otherwise;  we use the operator $\oplus : {\cal K} \times \{0,1\} \rightarrow {\cal K}$ such that $k \oplus b = -1$ if $k = -1$, $k \oplus b = k+b$ if $\{ k \neq -1,\top$ and $k+b \leq K\}$,  and $k \oplus b = \top$ in all other cases. The automaton $\det({\cal A}^w,K)$ has the following properties:

\begin{proposition} \label{prop:det}
Let $K \in \nat$ and $\langle {\cal A}^w,K \rangle$ be an \energy\UCWK. Then $\det({\cal A}^w,K)$ is a deterministic and complete energy automaton such that $L_{\uca,0}(\det({\cal A}^w,K),c_0) = L_{\uca,K}({\cal A}^w,c_0)$ for all $c_0 \in\nat$.
\qed\end{proposition}

\paragraph{Energy \UCWK and \LTLE realizability --} We now go through a series of results in a way to construct a safety game from a \UCW $\cal A$ such that $L_{\uca}({\cal A}) = \ \sem{\phi}$ (see Theorem~\ref{thm:reductionS}).

\begin{proposition} \label{prop:AutEL-realizable}
Let $\phi$ be an \LTL formula over $\langle P,w \rangle$. Then there exists a \UCW $\cal A$ such that $L_{\uca}({\cal A}) = \ \sem{\phi}$. Moreover, with the related \energy\UCW ${\cal A}^w$, the formula $\phi$ is $\sf{E}$-realizable with the initial credit $c_0$ iff there exists a Moore machine $\cal M$ such that $L({\cal M}) \subseteq L_{\uca}({\cal A}^w, c_0)$.
\end{proposition}

\begin{proof}
Let us take the negation $\neg \phi$ of $\phi$. It is well known (see for instance~\cite{DBLP:conf/focs/KupfermanV05}) that there exists a \NBW $\cal A$ such that $L_{\nba}({\cal A}) = \ \sem{\neg\phi}$; moreover $\overline{L_{\nba}({\cal A})}=L_{\uca}({\cal A})$ as the accepting conditions are dual. In this way we get the first part of the proposition. The second part follows from Proposition~\ref{prop:LTLEFiniteMem} and the definition of Moore machines.
\qed\end{proof}


\begin{theorem} \label{thm:UCWK}
Let $\phi$ be an \LTL formula over $\langle P, w_P \rangle$.  Let $\langle G_{\phi}, w, p \rangle$ be the associated energy parity game with $|S|$ being its the number of states, $d$ its number of priorities and $W$ its largest absolute weight. Let $\cal A$ be a \UCW with $n$ states such that $L_{\uca}({\cal A}) =\  \sem{\phi}$. Let $\mathbb{K} = 4 \cdot n \cdot |S|^2 \cdot d \cdot W$ and $\mathbb{C} = \mathbb{K}\cdot W$. Then $\phi$ is $\sf{E}$-realizable iff there exists a Moore machine $\cal M$ such that $L({\cal M}) \subseteq L_{\uca,\mathbb{K}}({\cal A}^w, \mathbb{C})$. 
\end{theorem}

\begin{proof}
By Propositions~\ref{prop:LTLEFiniteMem} and~\ref{prop:AutEL-realizable}, $\phi$ is $\sf{E}$-realizable for some initial credit $c_0$ iff  there exists a  Moore machine $\cal M$ such that $L({\cal M}) \subseteq L_{\uca}({\cal A}^w, c_0)$ and $\size{{\cal M}} = 4 \cdot |S|^2 \cdot d \cdot W$.  Consider now the product of $\cal M$ and ${\cal A}^w$: in any accessible cycle of this product, there is no 
accepting state of ${\cal A}^w$ (as shown similarly for the qualitative case \cite{FMSD11}) and the sum of the weights must be positive. The length of a path reaching such a 
cycle is at most $n\cdot \size{{\cal M}}$, therefore one gets $L({\cal M}) \subseteq L_{\uca,n\cdot\size{{\cal M}}}({\cal A}^w, n\cdot\size{{\cal M}}\cdot W )$.\qed
\end{proof}




\begin{theorem} \label{thm:reductionS}
Let $\phi$ be an \LTL formula. Then one can construct a safety game in which Player~1 has a winning strategy iff $\phi$ is {\sf E}-realizable.
\end{theorem}

\begin{proof}
Given $\phi$ an \LTL formula, let us describe the structure of the announced safety game $\langle G'_{\phi,\mathbb{K},\mathbb{C}},\alpha' \rangle$.
The involved constants $\mathbb{K}$ and $\mathbb{C}$ are those of Theorem~\ref{thm:UCWK}. By Theorem~\ref{thm:ESreductionS}, it is enough to show the statement with an energy safety game instead of a safety game. 
The construction of this energy safety game is very similar to the construction of a mean-payoff parity game from a deterministic parity automaton as given in the proof of Proposition~\ref{prop:reductionMPParity}. The main difference is that we will here use a \UCWK instead of a parity automaton.

Let $\phi$ be an \LTL formula and $\cal A$ be a \UCW such that $L_{\uca}({\cal A}) =\  \sem{\phi}$. By Theorem \ref{thm:UCWK} and Proposition~\ref{prop:det}, $\phi$ is $\sf{E}$-realizable iff there exists a Moore machine $\cal M$ such that $L({\cal M}) \subseteq L_{\uca,0}(\det({\cal A}^w,\mathbb{K}), \mathbb{C})$. Exactly as in Proposition~\ref{prop:reductionMPParity}, we derive from $\det({\cal A}^w,\mathbb{K}) = (\langle \Sigma_P, w\rangle, {\cal F},F_0,\beta,\Delta)$ a turn-based energy safety game $\langle G'_{\phi,\mathbb{K}}, w, \alpha' \rangle$ as follows. The construction of the graph and the definition of $w$ are the same, and $\alpha' =  {\cal F} \setminus \beta$. Similarly we have a bijection $\Theta : \Sigma_P^* \rightarrow \Pref(G'_{\phi,\mathbb{K}}) \cap (S_1S_2)^*S_1$ that can be extended to a bijection  $\Theta : \Sigma_P^{\omega} \rightarrow \Plays(G'_{\phi,\mathbb{K}})$. One can verify that for each $u \in \Sigma_P^{\omega}$, we have \emph{(i)} $\EL(u(n)) = \EL_{G'_{\phi,\mathbb{K}}}(\Theta(u(n)))$ for all $n \geq 0$, and \emph{(ii)} $\sum_{q \in \beta}\visit(\rho,q)=0$ for all runs $\rho$ on $u$ iff $\Theta(u)$ satisfies the objective $\safetyO_{G'_{\phi,\mathbb{K}}}(\alpha')$. It follows that $u \in L_{\uca,0}(\det({\cal A}^w,\mathbb{K}), \mathbb{C})$ iff $\Theta(u)$ satisfies the combined objective $\safetyO_{G'_{\phi,\mathbb{K}}}(\alpha') \cap \energyO_{G'_{\phi,\mathbb{K}}}(\mathbb{C})$~(*).

Suppose that $\phi$ is $\sf{E}$-realizable. By Theorem \ref{thm:UCWK}, there exists a 
finite-memory strategy $\lambda_O$ represented by a Moore machine $\cal M$ as given
before. As in the proof of Proposition~\ref{prop:reductionMPParity}, we use $\Theta$ to derive a strategy $\lambda_1$ of Player~1 that mimics $\lambda_O$. As $L({\cal M}) \subseteq L_{\uca,0}(\det({\cal A}^w,\mathbb{K}), \mathbb{C})$, by (*), it follows that $\lambda_1$ is winning in the energy safety game $\langle G'_{\phi,\mathbb{K}}, w, \alpha'  \rangle$ with the initial credit $\mathbb{C}$.

Conversely, suppose now that $\lambda_1$ is a winning strategy in  $\langle G'_{\phi,\mathbb{K}}, w, \alpha' \rangle$ with the initial credit $\mathbb{C}$. We again use $\Theta$ to derive a strategy $\lambda_O : (\Sigma_O\Sigma_I)^* \rightarrow \Sigma_O$ that mimics $\lambda_1$. As $\lambda_1$ is winning, by (*), we have $\outcome(\lambda_O) \subseteq L_{\uca,0}(\det({\cal A}^w,\mathbb{K}), \mathbb{C})$. It follows that $\phi$ is $\sf{E}$-realizable.%
%
%
\qed\end{proof}

A careful analysis of the complexity of the proposed Safraless procedure shows that it is in 2ExpTime. 


\section{Antichains}  \label{sec:Implementation}

In Section~\ref{sec:algorithm}, we have shown how to reduce the \LTLE realizability problem to a safety game. In this section we explain how to efficiently and symbolically solve this safety game with antichains.

\subsection{Description of the safety game}

In the proof of Theorem~\ref{thm:reductionS}, we have shown how to construct a safety game $\langle G'_{\phi,\mathbb{K},\mathbb{C}}, \alpha' \rangle$ from an \LTL formula $\phi$ such that $\phi$ is {\sf E}-realizable iff Player~1 has a winning strategy in this game. Let us give the precise construction of this game, but more generally for any values $K, C \in \nat$. Let ${\cal A} = (\langle \Sigma_P, w\rangle, Q, q_0, \final, \delta)$ be a  \UCW such that $L_{\uca}({\cal A}) = \ \sem{\phi}$, and $\det({\cal A}^w,K) = (\langle \Sigma_P, w\rangle, {\cal F},F_0,\beta,\Delta)$ be the related energy deterministic automaton. Let ${\cal C} = \{\bot,0,1, \ldots , C\}$. The turned-based safety game $\langle G'_{\phi,K,C}, \alpha' \rangle$ with $G'_{\phi,K,C} = (S = S_1 \uplus S_2, s_0, E)$ has the following structure:
\begin{itemize}
\item $S_1 = \{ (F,i,c) \mid F \in {\cal F}, i \in \Sigma_I, c \in {\cal C}\}$
\item $S_2 = \{ (F,o,c) \mid F \in {\cal F}, o \in \Sigma_O, c \in {\cal C}\}$
\item $s_0 = (F_0,j_0, C)$ with $j_0$ be an arbitrary symbol of $\Sigma_I$, and $F_0$ be the initial state of $\det({\cal A}^w,\mathbb{K})$
\item For all $\Delta(F,o \cup i) = F'$, $j \in \Sigma_I$ and $c \in {\cal C}$, the set $E$ contains the edges $$((F,j,c),(F,o,c')) \text{ and } ((F,o,c'),(F',i,c''))$$
where $c' = c \oplus w(o)$ and $c'' = c' \oplus w(i)$
\item $\alpha' = (S_1 \uplus S_2) \setminus \{ (F,\sigma,c) \mid \exists q, F(q) = \top  \text{ or } c = \bot \}$
\end{itemize}

Notice that given $K_1 \leq K_2$ and $C_1 \leq C_2$, if Player~1 has a winning strategy in the safety game $\langle G'_{\phi,K_1,C_1}, \alpha' \rangle$, then he has a winning strategy in the safety game $\langle G'_{\phi,K_2,C_2}, \alpha' \rangle$. The next corollary of Theorem~\ref{thm:reductionS} holds. 

\begin{corollary}
Let $\phi$ be an \LTL formula, and $K, C \in \nat$. If Player~1 has a winning strategy in the safety game $\langle G'_{\phi,K,C}, \alpha' \rangle$, then $\phi$ is $\sf{E}$-realizable.
\qed \end{corollary}

This property indicates that testing whether $\phi$ is $\sf{E}$-realizable can be done \emph{incrementally} by solving the family of safety games $\langle G'_{\phi,K,C}, \alpha' \rangle$ with increasing values of $K, C \geq 0$ until either Player~1 has a winning strategy in $\langle G'_{\phi,K,C}, \alpha' \rangle$ for some $K, C$ such that $0 \leq K  \leq \mathbb{K}$, $0 \leq C \leq \mathbb{C}$, or Player~1 has no winning strategy in $\langle G'_{\phi,\mathbb{K},\mathbb{C}}, \alpha' \rangle$.

\subsection{Antichains}

The \LTLE realizability problem can be reduced to a family of safety games $\langle G'_{\phi,K,C}, \alpha' \rangle$ with $0 \leq K  \leq \mathbb{K}$, $0 \leq C \leq \mathbb{C}$. We here show how to make more efficient the fixpoint algorithm to check whether Player~1 has a winning strategy in $\langle G'_{\phi,K,C}, \alpha' \rangle$, by avoiding to explicitly construct this safety game. This is possible because the states of $G'_{\phi,K,C}$ can be partially ordered, and the sets manipulated by the fixpoint algorithm can be compactly represented by the antichain of their maximal elements. 

\paragraph{Partial order and antichains --} 

Consider the safety game $\langle G'_{\phi,K,C}, \alpha' \rangle$ with $G'_{\phi,K,C} = (S = S_1 \uplus S_2, s_0, E)$ as defined above. We define the relation $\preceq$ $\subseteq$ $S \times S$ by
\begin{equation*}
\begin{split}
	(F', \sigma, c') \preceq(F, \sigma, c)\textnormal{ iff }&(i)\ F' \leq F \textnormal{ and}\\ &(ii)\ c' \geq c
\end{split}
\end{equation*}
where $F, F' \in {\cal F}$, $\sigma \in \Sigma_P$, $c, c' \in {\cal C}$, and $F' \leq F$  iff $F'(q)\leq F(q)$ for all $q$.
It is clear that $\preceq$ is a partial order. Intuitively, if Player~$1$ can win the safety game from $(F, \sigma, c)$, then he can also win from all $(F', \sigma, c') \preceq (F, \sigma, c)$ as $(i)$ it is more difficult to avoid $\top$ from $F$ than from $F'$, and $(ii)$ the energy level is higher with $c'$ than with $c$. Formally, $\preceq$ is a game simulation relation in the terminology of \cite{Alur98}. The next lemma will be useful later.

\begin{lemma} \label{order_oplus}
\begin{itemize}
\item For all $F, F' \in \mathcal{F}$ such that $F' \leq F$ and $o \cup i \in \Sigma_P$, we have $\Delta(F',o \cup i) \leq \Delta(F,o \cup i)$.
\item For all $c, c' \in \mathcal{C}$ such that $c' \geq c$ and $k \in \integers$, we have $c' \oplus k$ $\geq c \oplus k$.
\end{itemize}
\end{lemma}

A set $L \subseteq S$ is \emph{closed for $\preceq$} if $\forall (F, \sigma, c) \in S,$ $\forall (F', \sigma, c') \preceq (F, \sigma, c),$ $(F', \sigma, c') \in L$. Let $L_1$ and $L_2$ be two closed sets, then $L_1 \cap L_2$ and $L_1 \cup L_2$ are closed. The \emph{closure} of a set $L \subseteq S$, denoted by $\downarrow$$L$, is the set $\downarrow$$L = \{(F', \sigma, c') \in S  \mid  \exists (F, \sigma, c) \in L, (F', \sigma, c') \preceq (F, \sigma, c)\}$. Note that if $L$ is closed, then $\downarrow$$L = L$. A set $L \subseteq S$ is an \emph{antichain} if all elements of $L$ are incomparable for $\preceq$. Let $L \subseteq S$, we denote by $\lceil L \rceil$ the antichain composed of the  maximal elements of $L$.
If $L$ is closed then $\downarrow$$\lceil L \rceil = L$, i.e. antichains are compact \emph{canonical representations} for closed sets. The next proposition indicates how to compute antichains with respect to the union and intersection operations~\cite{FMSD11}.


\begin{proposition} \label{prop:unionintersection}
Let $L_1, L_2 \subseteq S$ be two antichains. Then:
\begin{itemize}
\item $\downarrow$$L_1$ $\cup$ $\downarrow$$L_2 = \ \downarrow$$\lceil L_1 \cup L_2\rceil$
\item $\downarrow$$L_1$ $\cap$ $\downarrow$$L_2 = \ \downarrow$$\lceil L_1$ $\sqcap$ $L_2\rceil$ 

where $\lceil L_1$ $\sqcap$ $L_2\rceil = \{(F_1, \sigma, c_1)$ $\sqcap$ $(F_2, \sigma, c_2) \mid (F_1, \sigma, c_1) \in L_1,  (F_2, \sigma, c_2) \in L_2 \}$, and $(F_1, \sigma, c_1)$ $\sqcap$ $(F_2, \sigma, c_2) : (q \mapsto \min(F_1(q), F_2(q)), \sigma, \max(c_1, c_2))$.
\end{itemize}
\end{proposition}

\paragraph{Fixpoint algorithm with antichains --} 

We recall the fixpoint algorithm to check whether Player~1 has a winning strategy in the safety game $\langle G'_{\phi,K,C}, \alpha' \rangle$. This algorithm computes the fixpoint $\Win_1(\alpha')$ of the sequence $W_0 = \alpha'$, $W_{k+1} = W_k \cap \{ \{s \in S_1 \mid \exists (s,s') \in E, s' \in W_k \} \cup \{s \in S_2 \mid \forall (s,s') \in E, s' \in W_k \} \}$ for all $k \geq 0$. Player~1 has a winning strategy iff $s_0 \in \Win_1(\alpha')$. Let us show how to efficiently implement this algorithm with antichains.

Given $L \subseteq S$, let us denote by ${\sf CPre}_1(L)$ the set $\{s \in S_1 \mid \exists (s,s') \in E, s' \in L \}$ and by ${\sf CPre}_2(L)$ the set $\{s \in S_2 \mid \forall (s,s') \in E, s' \in L \}$. We have the next lemma.

\begin{lemma} \label{lem:closed}
If $L \subseteq S$ is a closed set, then ${\sf CPre}_1(L)$ and ${\sf CPre}_2(L)$ are also closed.
\end{lemma}

\begin{proof}
To get the required property, it is enough to prove that if $(s,s') \in E$ and $r \preceq s$, then there exists $(r, r') \in E$ with $r' \preceq s'$. 

Let us first suppose that $(s,s') \in S_1 \times S_2$. Thus, by definition of $G'_{\phi,K,C}$, we have $s = (F,j,c)$ and $s' = (F,o,c \oplus w(o) )$ for some $F \in {\cal F}$, $c \in {\cal C}$, $j \in\Sigma_I$ and $o \in \Sigma_O$. Let $r = (G,j,d) \preceq (F,j,c)$, i.e. $G \leq F$ and $d \geq c$. We define $r' = (G,o, d \oplus w(o))$. Then $(r,r')  \in E$, $G \leq F$ and $d \oplus w(o) \geq c \oplus w(o)$ by Lemma~\ref{order_oplus}. It follows that $r' \preceq s' $. 

Let us now suppose that $(s,s') \in S_2 \times S_1$. We now have $s = (F,o,c)$ and $s' = (\Delta(F,o \cup i), i, c \oplus w(i))$. Let $r = (G,o,d) \preceq (F,o,c)$, and let us define $r' = (\Delta(G,o\cup i),i, d \oplus w(i))$. By Lemma~\ref{order_oplus}, we get $\Delta(G,o\cup i) \leq \Delta(F,o \cup i)$ and $d \oplus w(i) \geq c \oplus w(i)$. Therefore $r' \preceq s'$.
\qed\end{proof}

Notice that in the safety game $\langle G'_{\phi,K,C}, \alpha' \rangle$, the set $\alpha'$ is closed by definition. Therefore, by the previous lemma and Proposition~\ref{prop:unionintersection}, the sets $W_{k}$ computed by the fixpoint algorithm are closed for all $k \geq 0$, and can thus be compactly represented by their antichain $\lceil W_k \rceil$. Let us show how to manipulate those sets efficiently. For this purpose, let us consider in more details the following sets of predecessors for each $o \in \Sigma_O$ and $i \in \Sigma_I$: 
\begin{eqnarray*}
{\sf Pre}_o(L) &=& \{s \in S_1 \mid (s,s') \in E \text{ and } s' = (F,o,c) \in L, \text{ for some } F \in {\cal F}, c \in {\cal C} \} \\
{\sf Pre}_i(L) &=& \{s \in S_2 \mid (s,s') \in E \text{ and } s' = (F,i,c) \in L, \text{ for some } F \in {\cal F}, c \in {\cal C} \}
\end{eqnarray*} 
Notice that ${\sf CPre}_1(L) = \cup_{o \in \Sigma_O} {\sf Pre}_o(L)$ and ${\sf CPre}_2(L) = \cap_{i \in \Sigma_I} {\sf Pre}_i(L)$. Given $(F,o,c) \in S_2$ and $(F,i,c) \in S_1$, we define
\begin{eqnarray*}
\Omega(F,o,c) &=& \begin{cases} \{(F,i,c') \mid  i \in \Sigma_I,  c' = \min \{d \in {\cal C} \mid d \oplus w(o) \geq c \} \} \quad \text{ if } c' \text{ exists }  \\ \emptyset \quad  \text{ otherwise. } \end{cases} \\
\Omega(F,i,c) &=& \begin{cases} \{(F',o,c') \mid \begin{array}[t]{l} o \in \Sigma_O,  F' = \max\{G  \in {\cal F} \mid \Delta(G,o \cup i) \leq F\}, \\ c' = \min \{d \in {\cal C}  \mid d \oplus w(i) \geq c \} \}  \quad\text{ if } c' \text{ exists } \end{array} \\ \emptyset \quad \text{ otherwise. } \end{cases}
\end{eqnarray*} 
When defining the set $\Omega(F,\sigma,c)$, we focus on the worse predecessors with respect to the partial order $\preceq$. In this definition, $c'$ may not exist since the set $\{d  \in {\cal C}  \mid d \oplus w(\sigma) \geq c\}$ can be empty. However the set $\{G  \in {\cal F} \mid \Delta(G,o \cup i) \leq F\}$ always contains $G : q \mapsto -1$. Moreover, even if $\preceq$ is a partial order, $\max\{G  \in {\cal F} \mid \Delta(G,o \cup i) \leq F\}$ is unique. Indeed if $\Delta(G_1,o \cup i) \leq F$ and $\Delta(G_2,o \cup i) \leq F$, then $\Delta(G,o \cup i) \leq F$ with $G : q \mapsto \max(G_1(q), G_2(q))$. 

\begin{proposition} \label{prop:omega}
For all $F \in {\cal F}$, $\sigma \in \Sigma_P$ and $c \in {\cal C}$, 
${\sf Pre}_{\sigma}(\downarrow$$(F,\sigma,c)) = \downarrow$$\Omega(F,\sigma,c)$.
\end{proposition}
 
\begin{proof}
We only give the proof for $\sigma = i \in \Sigma_I$ since the case $\sigma \in \Sigma_O$ is a particular case. We prove the two following inclusions.
\begin{enumerate}[1)]
\item ${\sf Pre}_i(\downarrow$$(F,i, c)) \subseteq$ $ \downarrow$$\Omega(F,i,c)$\\
Let $s' = (G,i,d) \preceq (F,i,c)$ and $s = (G',o,d')$ such that $(s,s') \in E$. We have to show that $s \preceq \Omega(F,i,c)$. As $(s,s') \in E$, we have $\Delta(G', o \cup i) = G \leq F$ and $d' \oplus w(i) = d \geq c$. It follows that $(G',o,d') \preceq \Omega(F,i,c)$ by definition of $\Omega(F,i,c)$. 
\item $\downarrow$$\Omega(F,i,c) \subseteq {\sf Pre}_i(\downarrow$$(F,i, c))$\\
Let $(F',o,c') \in \Omega(F,i,c)$ and $s = (G',o,d') \preceq (F',o,c')$. We have to show that there exists $(s,s') \in E$ with $s' \preceq (F,i,c)$. By definition of $\Omega(F,i,c)$, we have $\Delta(F',o \cup i) \leq F$ and $c' \oplus w(i) \geq c$. As $G' \leq F'$ and $d' \geq c'$, it follows that $\Delta(G', o \cup i) \leq \Delta(F', o \cup i) \leq F$ and $d' \oplus w(i) \geq c' \oplus w(i) \geq c$ by Lemma~\ref{order_oplus}. Therefore with $s' = (\Delta(G', o \cup i), i, d' \oplus w(i))$, we have $(s,s') \in E$ and $s' \preceq (F,i,c)$. Thus $(G',o,d') \in {\sf Pre}_i(\downarrow$$(F,i, c))$.
\end{enumerate}
\qed\end{proof}

Propositions \ref{prop:unionintersection} and~\ref{prop:omega} indicate how to limit to antichains the computation steps of the fixpoint algorithm. 

\begin{corollary}
If $L \subseteq S$ is an antichain, then ${\sf CPre}_1(L) = \bigcup_{o \in \Sigma_O} \bigcup_{(F,o,c) \in L} \downarrow$$\Omega(F,o,c)$ and ${\sf CPre}_2(L) = \bigcap_{i \in \Sigma_I} \bigcup_{(F,i,c) \in L} \downarrow$$\Omega(F,i,c)$.
\end{corollary}

\paragraph{Optimizations --} 
The definition of $\Omega(F,i,c)$ requires to compute $\max\{G  \in {\cal F} \mid \Delta(G,o \cup i) \leq F\}$. This computation can be done more efficiently using the operator $\ominus : {\cal K} \times \{0,1\} \rightarrow {\cal K}$ defined as follows: $k \ominus b = \top$ if $k = \top$, $k \ominus b = -1$ if $k \neq \top, k - b \leq -1$, and $k \ominus b = k-b$ in all other cases. Indeed, using Lemma 4 of \cite{FMSD11}, one can see that 
$$\max\{G  \in {\cal F} \mid \Delta(G,o \cup i) \leq F\} = q \mapsto \min \{F(q') \ominus (q' \in \alpha') \mid q' \in \delta(q,o \cup i)\}.$$

It is possible to reduce the size of the safety game $\langle G'_{\phi,K,C}, \alpha' \rangle$ such that $S_1 =  \{ (F,c) \mid F \in {\cal F}, c \in {\cal C}\}$ instead of $\{ (F,i,c) \mid F \in {\cal F}, i \in \Sigma_I, c \in {\cal C}\}$. We refer to the proof of Proposition~\ref{prop:reductionMPParity} and footnote~\ref{fnsmaller} for the justification.

\paragraph{\LTLE-synthesis --}
If Player~1 has a winning strategy in the safety game $\langle G'_{\phi,K,C}, \alpha' \rangle$, that is, the given formula $\phi$ is {\sf E}-realizable, then it is easy to contruct a Moore machine ${\cal M}$ that realizes it. As described in \cite{FMSD11}, $\cal M$ can be constructed from the antichain $\lceil \Win_1(\alpha') \rceil$ computed by the fixpoint algorithm, with the advantage of having a small size bounded by the size of $\lceil \Win_1(\alpha') \rceil$ (see Section~\ref{subsec:experiments}).

\paragraph{Forward algorithm --}
The proposed fixpoint algorithm works in a backward manner. In~\cite{FMSD11}, the authors propose a variant of the {\sf OTFUR} algorithm of~\cite{CassezDFLL05} that computes a winning strategy for Player~1 (if it exists) in a forward fashion, starting from the initial state of the safety game.
This forward algorithm can be adapted to the safety game $\langle G'_{\phi,K,C}, \alpha' \rangle$. As for the backward fixpoint algorithm, it is not necessary to construct the game explicitly and antichains can again be used~\cite{FMSD11}. 
Compared to the backward algorithm, the forward algorithm has the following advantage: it only computes winning states $(F,\sigma,c)$ (for Player~$1$) which are reachable from the initial state. Nevertheless, it computes a single winning strategy if it exists, whereas the backward algorithm computes 
a fixpoint from which we can easily enumerate the set of all winning strategies (in the safety game).

\section{Extension to multi-dimensional weights} \label{sec:multi-dim}

\paragraph{Multi-dimensional \LTLMP and \LTLE realizability problems --}
The \LTLMP and \LTLE realizability problems can be naturally extended to multi-dimensional weights. Given $P$, we define a weight function $w : \Lit(P) \rightarrow \integers^m$, for some dimension $m \geq 1$. The concepts of energy level $\EL$, mean-payoff value $\MP$, and value $\Val$ are defined similarly. Given an \LTL formula $\phi$ over $\langle P, w \rangle$ and a threshold $\nu \in \rationals^m$, the \emph{multi-dimensional \LTLMP realizability problem (under finite memory)} asks to decide whether there exists a (finite-memory) strategy $\lambda_O$ of Player~$O$ such that $\Val(\Outcome(\lambda_O,\lambda_I)) \geq \nu$ against all strategies $\lambda_I$ of Player~$I$.\footnote{With $a \geq b$, we mean $a_i \geq b_i$ for all $i$, $1 \leq i \leq m$.} The \emph{multi-dimensional \LTLE realizability problem} asks to decide whether there exists a strategy $\lambda_O$ of Player~$O$ and an initial credit $c_0 \in \nat^m$ such that for all strategies $\lambda_I$ of Player~$I$, \emph{(i)} $u = \Outcome(\lambda_O,\lambda_I)\models \phi$, \emph{(ii)} $\forall n \geq 0,\ c_0 +  \EL(u(n)) \geq (0,\ldots,0)$.

\paragraph{Computational complexity --}
The 2ExpTime-completeness of the \LTLMP and \LTLE realizability problems have been stated in Theorem~\ref{thm:complexity} and~\ref{thm:complexityEnergy} in one dimension. In the multi-dimensional case, we have the next result.

\begin{theorem} \label{thm:complexityMulti}
The multi-dimensional \LTLMP realizability problem under finite memory and the multi-dimensional \LTLE realizability problem are in co-N2ExpTime.
\end{theorem}

Before giving the proof, we need to introduce \emph{multi-mean-payoff games} and \emph{multi-energy games} and some related results. Those games are defined as in Section~\ref{subsec:gamegraphs} with the only difference that the weight function $w$ assigns an $m$-tuple of weights to each edge of the underlying graph. The next proposition extends Proposition~\ref{prop:Doyen} to multiple dimensions.

\begin{proposition} [\cite{DBLP:conf/fsttcs/ChatterjeeDHR10}]  \label{prop:energyMPGeneralised}
Let $\langle G, w, m \rangle$ be a game with $w : E \rightarrow \integers^m$. Let $\nu \in \rationals^m$ be a threshold and $\lambda_1$ be a finite-memory strategy for Player~1. Then $\lambda_1$ is winning in the multi-mean-payoff parity game $\langle G, w, p \rangle$ with threshold $\nu$ iff $\lambda_1$ is winning in the multi-energy parity game $\langle G, w-\nu \rangle$ for some initial credit $c_0 \in \nat^m$.
\end{proposition}

In \cite{CRR12}, the authors study the initial credit problem for multi-energy parity games $\langle G, w, p, m \rangle$. They introduce the notion of self-covering tree\footnote{See  \cite{CRR12} for the definition and results.} associated with the game, and show that its 
depth is bounded by a constant $l = 2^{(h-1)\cdot |S|} \cdot (W \cdot |S| + 1)^{c\cdot m^2}$ where $|S|$ is the number of states of $G$, $W$ is its largest absolute weight, $h$ is the highest number of outgoing edges on any state of $S$, $m$ is the dimension, and $c$ is a constant independent of the game.  The next proposition states that 
multi-energy parity games reduce to multi-energy games.

\begin{proposition} [\cite{CRR12}] \label{prop:JFRandourGeneralized}
Let $\langle G, w, p, m \rangle$ be a multi-energy parity game with a priority function $p : S \rightarrow \{0,1, \ldots,  2 \cdot d \}$ and a self-covering tree of depth bounded by $l$. Then one can construct a multi-energy game $\langle G, w', m' \rangle$ with $m' = m + d$ dimensions and a largest absolute weight $W'$ bounded by $l$, such that a strategy is winning for Player~1 in $\langle G, w, p, m \rangle$ iff it is winning in $\langle G, w', m' \rangle$.
\end{proposition}

The next results are taken from~\cite{DBLP:conf/fsttcs/ChatterjeeDHR10} and~\cite{CRR12}.

\begin{theorem}[\cite{DBLP:conf/fsttcs/ChatterjeeDHR10,CRR12}] \label{thm:Generalized} 
\begin{itemize}
\item The initial credit problem for a multi-energy game is coNP-Complete.
\item If Player~1 has a winning strategy for the initial credit problem in a multi-energy parity game, then he can win with a finite-memory strategy of at most exponential size. 
\item Let $\langle G, w, m \rangle$ be a multi-energy game with a self-covering tree of depth bounded by $l$. If Player~1 has a winning strategy for the initial credit problem, then he can win with an initial credit $(C, \ldots, C) \in \nat^m$ such that $C = 2 \cdot l \cdot W$.\footnote{This result is extended in~\cite{CRR12} to multi-energy parity games thanks to Proposition~\ref{prop:JFRandourGeneralized}.}
\end{itemize}
\end{theorem}

\begin{proof}[of Theorem~\ref{thm:complexityMulti}]
We proceed as in the proof of the theorems~\ref{thm:complexity} and~\ref{thm:complexityEnergy} by reducing the \LTLMP\ (resp. \LTLE) realizability of formula $\phi$ to a multi-mean-payoff (resp. multi-energy) parity game $\langle G_{\phi}, w, p, m \rangle$. By Proposition~\ref{prop:energyMPGeneralised}, it is enough to study the multi-dimensional \LTLE realizability problem. We reduce the multi-energy parity game $\langle G_{\phi}, w, p, m \rangle$ to a multi-energy game $\langle G_{\phi}, w',  m' \rangle$ as described in Proposition~\ref{prop:JFRandourGeneralized}. Careful computations show that the multi-dimensional \LTLE realizability problem is in co-N2ExpTime, by using Theorems~\ref{thm:LTLParity} and~\ref{thm:Generalized}.
\qed\end{proof}

Theorem~\ref{thm:complexityMulti} states the complexity of the \LTLMP realizability problem under finite memory. Notice that it is reasonable to ask for finite-memory (instead of any) winning strategies. Indeed, the previous proof indicates a reduction to a multi-mean-payoff game; winning strategies for Player~1 in such games require infinite memory in general;  however, if Player~1 has a winning strategy for threshold $\nu$, then he has a finite-memory one for threshold $\nu - \epsilon$ for all $\epsilon > 0$~\cite{DBLP:journals/corr/abs-1209-3234}.

\paragraph{Safraless algorithm --}
As done in one dimension, we can similarly show that the multi-dimensional \LTLE-realizability problem can be reduced to a safety game for which there exist symbolic antichain-based algorithms. The multi-dimensional \LTLMP-realizability problem under finite memory can be solved similarly thanks to Proposition~\ref{prop:energyMPGeneralised}.

\begin{theorem} \label{thm:ReductionSafetyMulti}
Let $\phi$ be an \LTL formula. Then one can construct a safety game in which Player~1 has a winning strategy iff $\phi$ is {\sf E}-realizable.
\end{theorem}

\begin{proof}
The proof is very similar to the one of Theorem~\ref{thm:reductionS}. We only indicate the differences. 

First, let $\langle G_{\phi}, w, p\rangle$ be the multi-energy parity game associated with $\phi$ as described in the proof of Theorem~\ref{thm:complexityMulti}. By Theorem~\ref{thm:Generalized} and Proposition~\ref{prop:LTLEFiniteMem} adapted to this multi-dimensional game, we know that if $\phi$ is ${\sf E}$-realizable, then Player~$O$ has a finite-memory strategy with a memory size $M$ that is at most exponential in the size of the game. 

Second, we need to work with multi-energy automata over a weighted alphabet $\langle \Sigma_P,w \rangle$, such that $w$ is a function over \Lit(P) that assigns $m$-tuples of weights instead of a single weight. 

Third, from a \UCW $\cal A$ with $n$ states such that $L_{\uca}({\cal A}) =\  \sem{\phi}$, we construct, similarly as in the one-dimensional case, a safety game $\langle G'_{\phi,\mathbb{K},(\mathbb{C},\ldots,\mathbb{C})}, \alpha' \rangle$ whose positions store a counter for each state of $\cal A$ and an energy level for each dimension. The constants $\mathbb{K}$ and $\mathbb{C}$ are defined differently from the one-dimensional case: $\mathbb{K} = n \cdot M$, and $\mathbb{C}$ is equal to the constant $C$ of Theorem~\ref{thm:Generalized}.
\qed\end{proof}

\paragraph{Antichain-based algorithms --}
Similarly to the one-dimensional case, testing whether an \LTL formula $\phi$ is $\sf{E}$-realizable can be done incrementally by solving a family of safety games related to the safety game given in Theorem~\ref{thm:ReductionSafetyMulti}. These games can be symbolically solved by the antichain-based backward and forward algorithms described in Section~\ref{sec:Implementation}.

\section{Experiments}\label{sec:exp}

In the previous sections, in one or several dimensions, we have shown how to reduce the \LTLMP under finite memory and \LTLE realizability problems to a safety games, and how to derive symbolic antichain-based algorithms. This approach has been implemented in our tool {\sf Acacia+}. We briefly present this tool and give some experimental results.

\subsection{Tool {\sf Acacia+}}
In~\cite{DBLP:conf/cav/BohyBFJR12}, we present \textsf{Acacia+}, a tool for \LTL synthesis using antichain-based algorithms. The main advantage of this tool, in comparison with other \LTL synthesis tools, is to generate compact strategies that are easily usable in practice. This aspect can be very useful in many application scenarios like synthesis of control code from high-level LTL specifications, debugging of unrealizable \LTL specifications by inspecting compact counter strategies, and generation of small deterministic B\"uchi or parity automata from \LTL formulas (when they exist) \cite{DBLP:conf/cav/BohyBFJR12}.

\textsf{Acacia+} is now extended to the synthesis from \LTL specifications with mean-payoff objectives in the multi-dimensional setting. As explained in the previous sections, it solves incrementally a family of safety games, depending on some values $K$ and $C$, to test whether a given specification $\phi$ is \MP-realizable under finite memory. The tool takes as input an \LTL formula $\phi$ with a partition of its set $P$ of atomic signals, a weight function $w: \textsf{Lit}(P) \mapsto \mathbb{Z}^m$,
a threshold value $\nu \in \mathbb{Q}^m$, and two bounds $K\in \mathbb{Z}$ and $C\in \mathbb{Z}^m$ (the user can specify additional parameters to define the incremental policy). It then searches for a finite-memory winning strategy for Player~$O$, within the bounds of $K$ and $C$, and outputs a Moore machine if such a strategy exists. The last version of \textsf{Acacia+} can be downloaded at \url{http://lit2.ulb.ac.be/acaciaplus/} and it can also be used directly online via a web interface. Moreover, many benchmarks and results tables are available on the website.

\subsection{Experiments} \label{subsec:experiments}
In this section, we present some experiments. They have been done on a Linux platform with a 3.2GHz CPU (Intel Core i7) and 12GB of memory. 

\paragraph{Approaching the optimal value --} Let us come back to Example~\ref{ex:LTLMP}, where we have given a specification $\phi$ together with a 1-dimensional mean-payoff objective. For the optimal value $\nu_{\phi}$, we have shown that no finite-memory strategy exists, but finite-memory $\epsilon$-optimal strategies exist for all $\epsilon > 0$. In Table~\ref{table:ex2results}, we present the experiments done for some values of $\nu_{\phi} - \epsilon$. 
\begin{table}[ht]
	\caption{\textsf{Acacia+} on the specification of Example~\ref{ex:LTLMP} with increasing threshold values. The column $\nu$ gives the threshold values, $K$ and $C$ the minimum values required to obtain a winning strategy, $|\mathcal{M}|$ the size of the Moore machine representing the strategy, $time$ the execution time (in seconds) and $mem$ the total memory usage (in megabytes). Note that the execution times are given for the forward algorithm applied to the safety game with values $K$ and $C$ (and not with smaller ones). }
	\label{table:ex2results}
 	\begin{center}
 		\begin{tabular}{|c||c|c|c|c|c|}
		\hline
	  	$\nu$ & $K$ & $C$ & $|\mathcal{M}|$ & \ time (s) \ & \ mem (MB) \ \\
	  	\hline 
	  	$-1.2$ & $4$ & $7$ & $5$ & $0.01$ & $9.75$\\
  		$-1.02$ & $49$ & $149$ & $50$ & $0.05$ & $9.88$\\
  		$-1.002$ & $499$ & $1499$ & $500$ & $0.34$ & $11.29$\\
		$-1.001$ & $999$ & $2999$ & $1000$ & $0.89$ & $12.58$\\
  		$-1.0002$ & $4999$ & $14999$ & $5000$ & $15.49$ & $30$\\
		$-1.0001$ & $9999$ & $29999$ & $10000$ & $59.24$ & $48.89$\\
		\ $-1.00005$ \ & \ $19999$ \ & \ $99999$ \ & \ $20000$ \ & $373$ & $86.68$\\
  		\hline
		\end{tabular}
	\end{center}
\vspace{-0.7cm}
\end{table} 
The output strategies for the system behave as follows: grant the second client ($|\mathcal{M}|-1$) times, then grant once client~$1$, and start over. Thus, the system almost always plays $g_2w_1$, except every $|\mathcal{M}|$ steps where he has to play $g_1w_2$. Obviously, these strategies are the smallest ones that ensure the corresponding threshold values. They can also be compactly represented by a two-state automaton with a counter that counts up to $|\mathcal{M}|$. 
%
With $\nu = -1.001$ of Table~\ref{table:ex2results}, let us emphasize the interest of using antichains in our algorithms. The underlying state space manipulated
by our symbolic algorithm is huge: since $K = 999$, $C = 2999$  and the number of automata states is $8$, the number of states is around $10^{27}$. 
However the fixpoint computed backwardly is represented by an antichain of size $2004$ only.

\paragraph{No unsollicited grants --} The major drawback of the strategies presented in Table~\ref{table:ex2results} is that many unsollicited grants might be sent since the server grants the resource access to the clients in a round-robin fashion (with a longer access for client~$2$ than for client~$1$) without taking care of actual requests made by the clients. It is possible to express
in \LTL the fact that no unsollicited grants occur, but it is cumbersome. Alternatively, the \LTLMP specification can be easily rewritten with a multi-dimensional mean-payoff objective to avoid those unsollicited grants, as shown in Example~\ref{ex:LTLMMP}.
\begin{example}\label{ex:LTLMMP}
We consider the client-server system of Examples~\ref{ex:LTL} and \ref{ex:LTLMP} with the additional requirement that the server does not send unsollicited grants. This property can be naturally expressed by keeping the inital \LTL specification $\phi$ and proposing a multi-dimensional mean-payoff objective as follows. A new dimension is added by client, such that a request (resp. grant) signal of client $i$ has a reward (resp. cost) of $1$ on his new dimension. More precisely, let $\phi$ and $P$ as in Example~\ref{ex:LTLMP}, we define $w : \Lit(P) \rightarrow \integers^3$ as the weight function such that $w(r_1) = (0,1,0)$, $w(r_2) = (0,0,1)$, $w(g_1) = (0,-1,0)$, $w(g_2) = (0,0,-2)$, $w(w_1) = (-1,0,0)$, $w(w_2) = (-2,0,0)$ and $w(l) = (0,0,0)$, $\forall l  \in \Lit(P) \setminus \{r_1,r_2,g_1,g_2,w_1,w_2\}$.


For threshold $\nu = (-1,0,0)$, there is no hope to have a finite-memory strategy (see Example~\ref{ex:LTLMP}). For threshold $\nu = (-1.2,0,0)$ and values $K=4$, $C=(7,1,1)$, \textsf{Acacia+} outputs a finite-memory strategy computed by the backward algorithm, as depicted in Figure~\ref{fig:strategy}.
In this figure, the strategy is represented by a transition system where the red state is the initial state, and the transitions are labeled with symbols $o | i$ with $o \in \Sigma_O$ and $i \in \Sigma_I$. Notice that the labels of all outgoing transitions of a state share the same $o$ part (since we deal with a strategy). This transition system can be seen as a Moore machine $({\cal M}, m_0, \alpha_U, \alpha_N)$ with the same state space (the set $M$ of memory states), and such that for each transition from $m$ to $m'$ labeled by $o \cup i$, we have $\alpha_U(m,i) = m'$ and $\alpha_N(m) = o$. We can verify that no unsollicited grant is done if the server plays according to this strategy. Moreover, this is the smallest strategy to ensure a threshold of $(-1.2,0,0)$ against the most demanding behavior of the clients, i.e. when they both make requests all the time (see states $3$ to $7$), and that avoid unsollicited grants against any other behaviors of the clients (see states $0$ to $2$).

\begin{figure}[!p]
\begin{center}
 \includegraphics[width=0.60\textwidth]{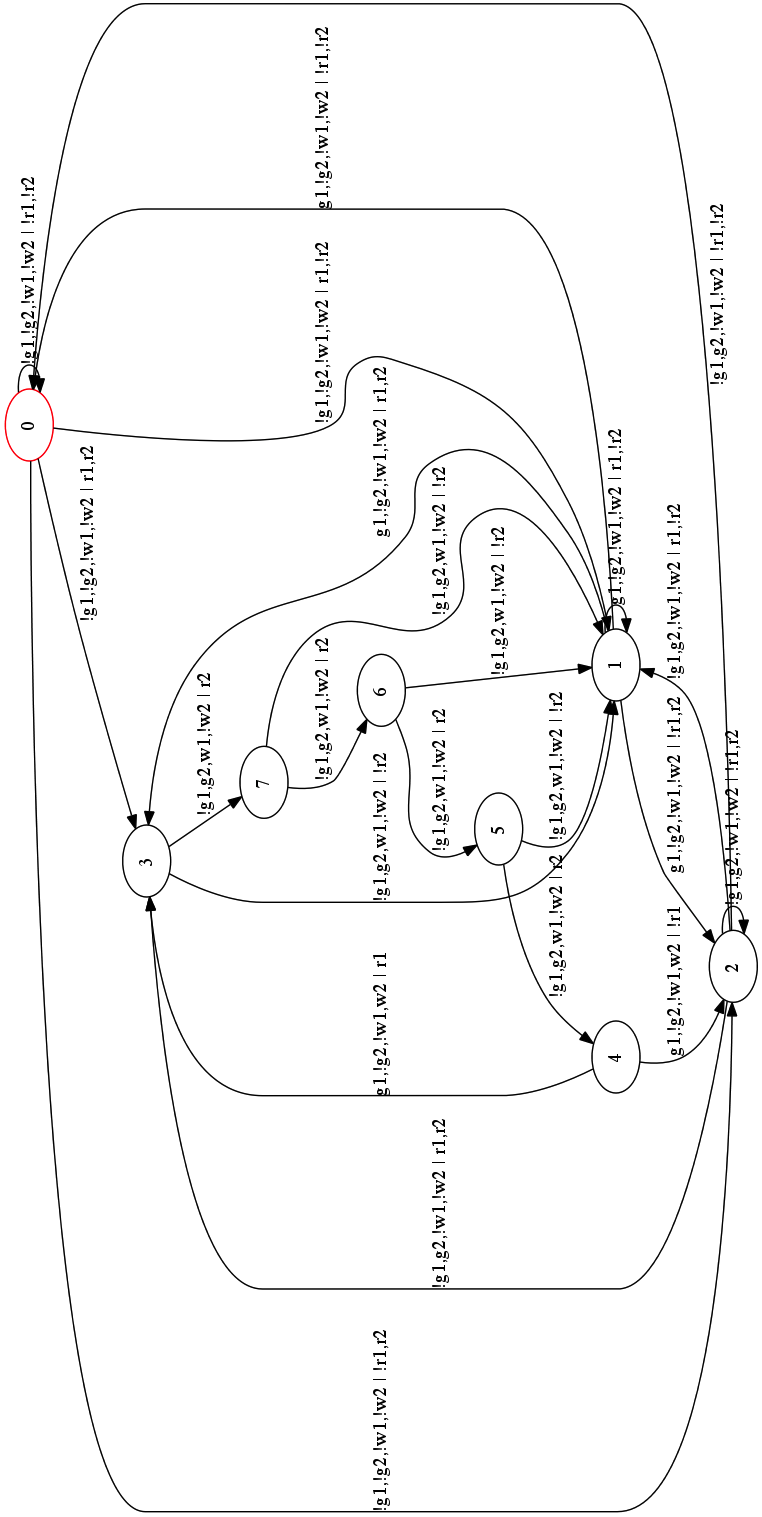}
	\caption{\label{fig:strategy} Strategy output by \textsf{Acacia+} for the specification of Example~\ref{ex:LTLMMP}, threshold $\nu = (-1.2,0,0)$ and values $K=4$, $C=(7,1,1)$, using the backward algorithm}
\end{center}
\end{figure}
\end{example}

From Example~\ref{ex:LTLMMP}, we derive a benchmark of multi-dimensional examples parameterized by the number of clients making requests to the server. Some experimental results of \textsf{Acacia+} on this benchmark are synthetized in Table~\ref{table:SRA}. 

\begin{table}[ht]
	\caption{\textsf{Acacia+} on the Shared Resource Arbiter benchmark parameterized by the number of clients, with the forward algorithm. The column $c$ gives the number of clients, $\nu$ the threshold, $K$ (resp. $C$) the minimum value (resp. vector) required to obtain a winning strategy, $|\mathcal{M}|$ the size of the Moore machine representing the strategy, $time$ the total execution time (in seconds) and $mem$ the total memory usage (in megabytes).}
	\label{table:SRA}
	\begin{center}
 		\begin{tabular}{|c|c||c|c|c|c|c|}
		\hline
	  	 $c$ & $\nu$ & $K$ & $C$ & $|\mathcal{M}|$ & \ time (s) \ & \ mem (MB) \  \\
	  	\hline 
	  	$2$ & $(-1.2,0,0)$ & $4$ & $(7,1,1)$ & $11$ & $0.02$  &  $10.04$\\
  		$3$ & $(-2.2,0,0,0)$ & $9$ & $(19,1,1,1)$ & $27$ & $0.22$ & $10.05$\\
  		$4$ & $(-3.2,0,0,0,0)$ & $14$ & $(12,1,1,1,1)$ & $65$ & $1.52$ & $12.18$\\
  		$5$ & $(-4.2,0,0,0,0,0)$ & $19$ & $(29,1,1,1,1,1)$ & $240$ & $48$ & $40.95$\\
  		\ $6$ \ & \ $(-5.2,0,0,0,0,0,0)$ \ & \ $24$ \ & \ $(17,1,1,1,1,1,1)$ \ & \ $1716$ \ & $3600$ &  $636$\\
		\hline
		\end{tabular}
	\end{center}
\vspace{-1cm}
\end{table}

\paragraph{Approching the Pareto curve --} As last experiment, we consider the 2-client \LTLMP specification of Example~\ref{ex:LTLMMP} where we split the first dimension of the weight function into two dimensions, such that $w(w_1) = (-1,0,0,0)$ and $w(w_2) = (0,-2,0,0)$. With this new specification, since we have several dimensions, there might be several optimal values for the pairwise order, corresponding to trade-offs between the two objectives that are $(i$) to quickly grant client~$1$ and $(ii)$ to quickly grant client~$2$. In this experiment, we are interested in approaching, by hand, the \emph{Pareto curve}, which consists of all those optimal values, i.e. to find finite-memory strategies that are incomparable w.r.t. the ensured thresholds, these thresholds being as large as possible. We give some such thresholds in Table~\ref{table:pareto}, along with minimum $K$ and $C$ and strategies size. It is difficult to automatize the construction of the Pareto curve. Indeed, \textsf{Acacia+} cannot test (in reasonable time) whether a formula is \MP-unrealizable for a given threshold, since it has to reach the huge theoretical bound on $K$ and $C$.  This raises two interesting questions that we let as future work: how to decide efficiently that
a formula is \MP-unrealizable for a given threshold, and how to compute points of the Pareto curve efficiently.

\begin{table}[ht]
	\caption{\textsf{Acacia+} to approach Pareto values. The column $\nu$ gives the threshold, relatively close to the Pareto curve, $K$ (resp. $C$) the minimum value (resp. vector) required to obtain a winning strategy, $|\mathcal{M}|$ the size of the Moore machine representing the strategy. }
	\label{table:pareto}
 	\begin{center}
 		\begin{tabular}{|c||c|c|c|}
		\hline
	  	 $\nu$ & $K$ & $C$ & $|\mathcal{M}|$ \\
	  	\hline 
	  	 $(-0.001, -2, 0, 0)$ & \ $999$ \ & \ $(1999,1,1,1)$ \  & \ $2001$ \ \\
  		 $(-0.15, -1.7, 0, 0)$ & $55$ & $(41,55,1,1)$ & $42$ \\
		$(-0.25, -1.5, 0, 0)$ & $3$ & $(7,9,1,1)$ & $9$ \\
		$(-0.5, -1, 0, 0)$ & $1$ & $(3,3,1,1)$ & $5$ \\
  		$(-0.75, -0.5, 0, 0)$ & $3$ & $(9,7,1,1)$ & $9$ \\
		\ $(-0.85, -0.3, 0, 0)$ \ & $42$ & $(55,41,1,1)$ & $9$ \\
  		 $(-1, -0.01, 0, 0)$ & $199$ & $(1,399,1,1)$ & $401$ \\
  		\hline
		\end{tabular}
	\end{center}
\end{table} 



\bibliographystyle{abbrv}
\bibliography{biblio}

\end{document}